  \documentclass[a4paper,12pt]{article}
\usepackage[utf8]{inputenc}

\usepackage{amssymb}
\usepackage{authblk}

\usepackage{graphicx}
\usepackage{graphics}
\usepackage{fancyhdr}
\usepackage{amsmath}
\usepackage{amsfonts}
\usepackage{amssymb}
\usepackage{amsthm}
\usepackage{epsfig}
\usepackage{lipsum} %
\usepackage{color}
\usepackage{amsthm}
\usepackage{sectsty} 
\usepackage{hyperref}
\usepackage{comment}
\usepackage{url}
\usepackage{empheq}
\usepackage{eurosym}
\usepackage{mathrsfs}
\usepackage[english]{babel}
\usepackage{booktabs}
\usepackage{caption}
\usepackage{subfig}
\usepackage{bm}
\usepackage{tikz}
\usepackage{xcolor}
\usepackage{color}
\usepackage{enumitem}
\usetikzlibrary{shapes,arrows}
\usepackage{listings}
\usepackage{rotating}

\usepackage[margin=1in]{geometry}

\theoremstyle{definition} \newtheorem{de}{Definition}[section]
\theoremstyle{plain}      \newtheorem{te}[de]{Theorem}
\theoremstyle{remark}     \newtheorem{os}[de]{Remark}
\theoremstyle{plain}      \newtheorem{pr}[de]{Proposition}
\theoremstyle{plain}      \newtheorem{lem}[de]{Lemma}
\theoremstyle{plain}	  \newtheorem{co}[de]{Corollary}
\theoremstyle{definition} 
\theoremstyle{remark}		
\theoremstyle{remark}		
\theoremstyle{plain}        \newtheorem*{beware*}{Beware}
\theoremstyle{plain}      \newtheorem{as}{Assumption}

\newcommand{\numberset}{\mathbb}
\newcommand{\N}{\numberset{N}}
\newcommand{\R}{\numberset{R}}

\newcommand{\p}{\numberset{P}}
\newcommand{\E}{\numberset{E}}
\newcommand{\T}{\numberset{T}}

\newcommand{\F}{\mathscr{F}}

\newcommand{\<}{\langle}

\newcommand{\Saff}{S_{t_k}^{\Xi}}

\usepackage[titletoc]{appendix}
\usepackage{titlesec}
\titlelabel{\thetitle.\quad}

\usepackage{multirow,array}
\newcommand*\rot{\rotatebox{90}}

\usepackage{setspace}
\usepackage{endnotes}
 \usepackage{apacite}
 \usepackage[round,sort&compress,comma]{natbib}
 \usepackage{doi}
 
\usepackage{pdfpages}
 \title{Instabilities in Multi-Asset and Multi-Agent \\Market Impact Games}

\author[a,\footnote{Corresponding author. ORCID: 0000-0003-0915-9002}]{ Francesco Cordoni}
\author[b,c]{ Fabrizio Lillo}

\affil[a]{Dipartimento di Economia e Management, Universit\`a di Pisa, \protect\\ Via C. Ridolfi, 10 - 56124 Pisa (PI), Italy\protect\\ E-mail: francesco.cordoni@sns.it \vspace{2.5pt}} 
\affil[b]{Scuola Normale Superiore, \protect\\ Piazza dei Cavalieri, 7 - 56126 Pisa (PI), Italy \vspace{2.5pt}}
\affil[c]{Dipartimento di Matematica, Universit\`a di Bologna, \protect\\ 
Piazza di Porta San Donato, 5 - 40126
Bologna (BO), Italy\protect\\ E-mail: fabrizio.lillo@unibo.it
} 

\date{\today}
\begin{document}


\maketitle
\begin{abstract}
    We consider the general problem of a set of agents trading a portfolio of assets in the presence of transient price impact and additional quadratic transaction costs and we study, with analytical and numerical methods, the resulting Nash equilibria. Extending significantly the framework of \cite{schied2018market} and \cite{luo_schied},
who considered the single asset case, we prove the existence and uniqueness of the corresponding Nash equilibria for the related mean-variance optimization problem. We then focus our attention on the conditions on the model parameters making the trading profile of the agents at equilibrium, and as a consequence the price trajectory, wildly oscillating and the market unstable. While \cite{schied2018market} and \cite{luo_schied} highlighted the importance of the value of transaction cost in determining the transition between a stable and an unstable phase, we show that also the scaling of market impact with the number of agents $J$  and the number of assets $M$ determines the asymptotic stability (in $J$ and $M$) of markets. 

\end{abstract}
\textbf{Keywords}: Market impact;
Game theory and Nash equilibria;
Transaction costs;
Market microstructure; High Frequency Trading; 
Cross-Impact.
\\

\noindent
\textbf{Statements and Declarations}: Not applicable.
\clearpage

\section{Introduction}

Instabilities in financial markets have always attracted the attention of researchers, policy makers and practitioners in the financial industry because of the role that financial crises have on the real economy. Despite this, a clear understanding of the sources of financial instabilities is still missing, in part probably because several origins exist and they are different at different time scales. The recent automation of the trading activity has raised many concerns about market instabilities occurring at short time scales (e.g. intraday), also because of the attention triggered by the Flash Crash of May 6th, 2010 (\cite{Kirilenko}) and the numerous other similar intraday instabilities observed in more recent years (\cite{Johnson}, \cite{golub}, \cite{Calcagnile}, \cite{Brogaard}), such as the Treasury bond flash crash of October 15th, 2014.
The role of High Frequency Traders (HFTs), Algo Trading, and market fragmentation in causing these events has been vigorously debated, both theoretically and empirically (\cite{golub}, \cite{Brogaard}).

One of the puzzling characteristics of market instabilities is that a large fraction of them appear to be endogenously generated, i.e. it is often very difficult to find an exogenous event (e.g. a news) which can be considered at the origin of the instability (\cite{Cutler,Fair,Joulin}). Liquidity plays a crucial role in explaining these events. Markets are, in fact, far from being perfectly elastic and any order or trade causes prices to move, which in turn leads to a cost (termed slippage) for the investor. The relation between orders and price is called market impact. In order to minimize market impact cost, when executing a large volume it is optimal for the investor to split the order in smaller parts which are executed incrementally over the day or even across multiple days. One of the origins of  market impact cost is predatory trading (\cite{Brunnermeier,Carlin}): the knowledge that a trader is purchasing progressively a certain amount of assets can be used to make profit by buying at the beginning and selling at the end of the trader's execution. Part of the core strategy of HFTs is exactly predatory trading. Now, the combined effect on price of the trading of the predator and of the prey can lead to large price oscillations and market instabilities. In any case, it is clear that the price dynamics is the result of the (dynamical) equilibrium between the activity of two or more agents simultaneously trading. 

This equilibrium can be studied by modeling the above setting as a market impact game (\cite{Carlin} , \cite{schoneborn2008trade}, \cite{Moallemi},  \cite{Lachapelle,schied2018market,Strehle,strehle2017single}). In a nutshell, in a market impact game, two traders want to trade the same asset in the same time interval. While trading, each agent modifies the price because of market impact, thus when two (or more) traders are simultaneously present, the optimal execution schedule of a trader should take into account the simultaneous presence of the other trader(s). As customary in these situations, the approach is to find the Nash equilibrium, which in general depends on the market impact model.

Market impact games are a perfect modeling setting to study endogenously generated market instabilities. A major step in this direction has been recently made\footnote{
Furthermore, many works have recently examined the continuous time setting, e.g., \cite{schied2017high}, 
\cite{Bayraktar} and mean field games approach, e.g., \cite{cardaliaguet2018mean}, \cite{Fu_mean_field}.} by \cite{schied2018market}.  
By using the transient impact model of \cite{Bouchaud, BFL} plus a quadratic temporary impact cost (which can alternatively be interpreted as a quadratic transaction cost, see below), they have recently considered a simple setting with two identical agents liquidating a single asset and derived the Nash equilibrium. Interestingly, they also derived analytically the conditions on the transaction cost under which the Nash equilibrium displays huge oscillations of the trading volume and, as a consequence, of the price, thus leading to market instabilities\footnote{In their paper, Schied and Zhang interpret the large alternations of buying and selling activity observed at instability as the ``hot potato game" among HFTs empirically observed during the Flash Crash \citep{CFTC,Kirilenko}.}. Specifically, they proved the existence of a sharp transition between stable and unstable markets at a specific value of the transaction cost parameter.

Although the paper of Schied and Zhang highlights a key mechanism leading to market instability, several important aspects are left unanswered. First, market instabilities rarely involve only one asset and, as observed for example during the Flash Crash, a cascade of instabilities affects very rapidly a large set of assets or the entire market (\cite{CFTC}). This is due to the fact that optimal execution strategies often involve a {\it portfolio} of assets rather than a single one (see, e.g. \cite{Gerry}). Commonality of liquidity across assets (\cite{Chordia} and cross-impact effects (\cite{alfonsi2016multivariate}, \cite{schneider2018cross}) makes the trading on one asset triggers price changes on other assets.
Furthermore,
\cite{cespa2014} show that a drop in liquidity in one asset can propagate in another asset, causing a market liquidity crash.
Thus, it is natural to ask: is a large market more or less prone to market instabilities? How does the structure of cross-impact and therefore of liquidity commonality affect the market stability? 
A second class of open questions regards instead the market participants. Do the presence of more agents simultaneously trading one asset tends to stabilize the market? While the solution of Schied and Zhang considers only two traders, it is important to know whether having more agents is beneficial or detrimental to market stability. For example, regulators and exchanges could implement mechanisms to favor or disincentive participation during turbulent periods. Answering this question requires solving the impact game with a generic number of agents and it is discussed in the single asset case in  \cite{luo_schied}. 

In this paper we extend considerably the setting of Schied and Zhang by answering the above research questions. Specifically, starting from \cite{luo_schied},
we consider (i) the case when agents trade multiple assets simultaneously and cross market impact is present and we provide explicit representations of related Nash equilibria; (ii) after studying how trading conditions may be affected by the cross impact, we derive theoretical results on market stability for the $J=2$ agents by showing
how it is related to cross-impact effects; (iii)
we study numerically market stability in the general case and we extend a previous result and conjecture of \cite{luo_schied} in the multi-asset case.

It is important to notice that in market impact games, market impact is taken as exogenously given. Market microstructure literature has extensively discussed its endogenous nature since the seminal work of \cite{Kyle}.  Theoretical and empirically studies have investigated and provided evidence of how market impact might depend on number of agents and of traded assets, e.g., \cite{Bagnoli}, 
\cite{benzaquen2017dissecting}, 
\cite{coimpact}, \cite{garcia2020multivariate}. Therefore we will consider this dependence and show how the stability of markets in market impact games depends on the way impact scales with the number of agents and the number of assets. We find that, if market impact is independent from the number of agents and assets, larger and more crowded markets are more prone to market instability. However, if, as observed empirically and proposed theoretically, market impact suitably scales with these two quantities, stability can be recovered.


The paper is organized as follows. In Section
\ref{sec_market_impact} we recall some notation of the market impact games framework and the \cite{luo_schied}
model. We extend the basic model of \cite{luo_schied} to the multi-asset case in Section \ref{se_MUlti_assets_case}, where we find the
corresponding Nash equilibria for different objective functions.
We analyse how the cross-impact modifies the trading
profile and trading conditions in Section \ref{cross_trading_conditions}.
Finally, in Section \ref{sec_inst_comm} we study 
how the cross-impact matrix affects
the market stability and we present how impact must scale with the number of assets to preserve stability. Finally, in Section 6 we draw somw conclusions. 

\section{Market Impact Games}\label{sec_market_impact}
Consider two traders who want to trade simultaneously a certain number of shares, minimizing the trading cost. 
Since the trading of one agent affects the price, the other agent must take into account the presence of the former in optimizing her execution. This problem is termed {\it market impact game} and has received considerable attention in recent years (\cite{Carlin,schoneborn2008trade,Moallemi, Lachapelle,schied2018market,Strehle,strehle2017single}). The seminal paper by Schied and Zhang,  (\cite{schied2018market}), considers a market impact game between two identical agents trading the same asset in a given time period. 

 When none of the two agents trade, the price dynamics is 
 described by 
 the so called unaffected price process $S_t^0$ which is 
 a right-continuous martingale defined 
 on a given probability space 
 $(\Omega,(\F_t)_{t\geq0},\F,\p)$.
 A trader wants to unwind
 a given initial position with inventory $X$, 
 where a positive (negative) inventory means a short (long) position,
 during a given trading time grid $ \T=\{t_0,t_1,\ldots,t_N\},$ where
 $0=t_0<t_1<\cdots <t_N=T$ and following 
 an admissible strategy, which is defined as follows: 

 \begin{de}[Admissible Strategy]\label{de_admi_strate_schied}
 Given $\T$ and $X$, an \emph{admissible trading strategy} for $\T$ and $X\in \R$ is a vector 
 $\bm{\zeta}=(\zeta_0,\zeta_1,\ldots,\zeta_N)$ of random variables such that: 
 \begin{itemize}
  \item $\zeta_k \in \F_{t_k} \text{ and bounded},\  \forall k=0,1,\ldots,N.$
  \item $\zeta_0+\zeta_1+\cdots+\zeta_N=X$.
 \end{itemize}
 \end{de} 

The random variable
$\zeta_k$ represents the order flow at trading time $t_k$ where positive (negative)
flow corresponds to a sell (buy) trade of volume $|\zeta_k|$. We denote with $X_1$ and $X_2$ the initial 
inventories of the two considered agents playing the game and with  $\Xi=(\xi_{i,k})\in \R^{2 \times (N+1)}$ the matrix 
of the respective strategies, where 
$\bm{\xi}_{1,\cdot}=\{ \xi_{1,k} \}_{k \in \numberset{T}}$
and 
$\bm{\xi}_{2,\cdot}=\{ \xi_{2,k} \}_{k \in \numberset{T}}$
are the strategies of trader $1$ and $2$, respectively.
Traders are subject to fees and transaction costs and their objective is to minimize them by optimizing the execution. As customary in the literature, the costs are modeled by two components. The first one is a temporary impact component modeled by a quadratic term $\theta \xi_{j,k}^2$, respectively for trader $j$, which does not affect the price dynamics and depends on the immediate liquidity present in the order book. Notice that, as discussed in \cite{schied2018market}, this term can also be interpreted as a quadratic transaction fee. Here we do not specify exactly what this term represents, sticking to the mathematical modeling approach of Schied and Zhang. 

The second component is related to permanent impact and affects future price dynamics. Following \cite{schied2018market}, we consider
the celebrated transient impact model of \cite{Bouchaud, BFL}, which describes the price process $S_{t}^{\Xi}$
affected by the 
strategies $\Xi$ of the two traders, i.e.,
\[
   S_t^{\Xi}= S_t^0 -\sum_{t_k <t} G(t-t_k)(\xi_{1,k}+\xi_{2,k}), 
   \quad \forall\ t \in \numberset{T},
\]
where $G:\R_{+}\to \R_+$ is the so called \emph{decay kernel}, which describes the lagged price impact of a unit buy or sell order
over time.
Usual assumptions on $G$ are satisfied, i.e., it is convex, nonincreasing, 
nonconstant so that
 $t\mapsto G(|t|)$ is strictly positive definite in the sense of Bochner,
 see \cite{Alfonsi} and \cite{schied2018market}. Notice that by choosing a constant kernel $G$, one recovers the celebrated Almgren-Chriss model (\cite{almgren2001optimal}).

The cost faced by each agent is the sum of the two components above. Specifically, 
let us denote with $\mathscr{X}(X,\T)$ the set of admissible strategies
for the initial inventory $X$ on a specified time grid $\T$, the cost functions are defined as:
  \begin{de}[\cite{schied2018market}]\label{de_cost} Given $\T=\{t_0,t_1,\ldots,t_N\}$, $X_1$ and $X_2$. Let $(\varepsilon_i)_{i=0,1,\ldots N}$
     be an i.i.d. sequence of Bernoulli $\left(\frac{1}{2}\right)$-distributed random variables 
     that are independent of $\sigma(\bigcup_{t\geq0}\F_t)$. Then the 
     \emph{cost of $\bm{\xi}_{1,\cdot}\in \mathscr{X}(X_1,\T)$ given $\bm{\xi}_{2,\cdot}\in \mathscr{X}(X_2,\T)$}
     is defined as
     \[
      C_{\T}(\bm{\xi}_{1,\cdot}|\bm{\xi}_{2,\cdot})=
     \sum_{k=0}^N \left(
      \frac{G(0)}{2}\xi_{1,k}^2-\Saff \xi_{1,k}+\varepsilon_k G(0) \xi_{1,k} \xi_{2,k} +
      \theta \xi_{1,k}^2
      \right)+ X_1 S_0^0
     \]
     and the \emph{costs of $\bm{\xi}_{2,\cdot}$ given $\bm{\xi}_{1,\cdot}$} are
          \[
      C_{\T}(\bm{\xi}_{2,\cdot}|\bm{\xi}_{1,\cdot})=
     \sum_{k=0}^N \left(
      \frac{G(0)}{2}\xi_{2,k}^2-\Saff \xi_{2,k}+(1-\varepsilon_k ) G(0) \xi_{1,k} \xi_{2,k} +
      \theta \xi_{2,k}^2
      \right)+ X_2 S_0^0.
     \]
    \end{de}
Thus the execution priority at time $t_k$ is given to the agent who wins an independent coin toss game, represented by a Bernoulli variable $\varepsilon_k$,
which is a fair game in the framework of \cite{schied2018market}.
  Given the time grid $\T=\{t_0,t_1,\ldots,t_N\}$ and the initial values 
  $X_1,X_2 \in \R$, we define the \emph{Nash Equilibrium}
  as a pair $(\bm{\xi}_{1,\cdot}^*,\bm{\xi}_{2,\cdot}^*)$ of strategies in $\mathscr{X}(X_1,\T)\times 
  \mathscr{X}(X_2,\T)$ such that
  \[
  \begin{split}
   \E[C_{\T}(\bm{\xi}_{1,\cdot}^*|\bm{\xi}_{2,\cdot}^*)]&=\min_{\bm{\xi}_{1,\cdot}\in \mathscr{X}(X_1,\T)}
   \E[C_{\T}(\bm{\xi}_{1,\cdot}|\bm{\xi}_{2,\cdot}^*)] \text{ and }\\
   \E[C_{\T}(\bm{\xi}_{2,\cdot}^*|\bm{\xi}_{1,\cdot}^*)]&=\min_{\bm{\xi}_{2,\cdot}\in \mathscr{X}(X_2,\T)}
   \E[C_{\T}(\bm{\xi}_{2,\cdot}|\bm{\xi}_{1,\cdot}^*)].
   \end{split}
  \]
One of main results of \cite{schied2018market} is the proof, under general assumptions,
of the existence and uniqueness of the Nash equilibrium.
Moreover, they showed that 
this equilibrium is deterministically given by a linear combination of two constant vectors, namely
\begin{align}\label{eq_1_S&Z}
\bm{\xi}_{1,\cdot}^*&=\frac{1}{2} (X_1 +X_2)\bm{v}
+\frac{1}{2} (X_1 -X_2)\bm{w}
\\
\label{eq_2_S&Z}
\bm{\xi}_{2,\cdot}^*&=\frac{1}{2} (X_1 +X_2)\bm{v}
-\frac{1}{2} (X_1-X_2)\bm{w},
\end{align}
where the fundamental solutions $\bm{v}$ and $\bm{w}$ are 
defined as 
\begin{align*}
   \bm{v}&=\frac{1}{\bm{e}^T (\Gamma_{\theta}+\widetilde{\Gamma})^{-1}\bm{e}}(\Gamma_{\theta}+\widetilde{\Gamma})^{-1}\bm{e}\\ 
     \bm{w}&=\frac{1}{\bm{e}^T (\Gamma_{\theta}-\widetilde{\Gamma})^{-1}\bm{e}}(\Gamma_{\theta}-\widetilde{\Gamma})^{-1}\bm{e}.
\end{align*}
and  $\bm{e}=(1,\ldots,1)^T \in \R^{N+1}$. The kernel matrix $\Gamma \in \R^{(N+1)\times (N+1)}$ is given by
  \[
   \Gamma_{ij}=G(|t_{i-1}-t_{j-1}|), \quad i,j=1,2,\ldots,N+1,
  \]
 and for $\theta\geq0$ it is $
   \Gamma_{\theta}:=\Gamma+2\theta I$,   and the matrix $\widetilde{\Gamma}$ is given by
  \[
\widetilde{\Gamma}_{ij}=
  \begin{cases}
   \Gamma_{ij} & \text{ if }i>j\\
      \frac{1}{2}G(0) & \text{ if }i=j,\\
      0 & \text{ otherwise.}
  \end{cases}
  \]
   As shown by \cite{schied2018market} all these matrices are positive definite.
   
An interesting result of \cite{schied2018market} concerns the stability
of the Nash equilibrium related to the transaction costs parameter $\theta$ and the decay kernel $G$. Generically, following \cite{schied2018market}, we say that a market is \emph{unstable} if the trading strategies at the Nash equilibrium exhibit spurious oscillations, i.e., if there exists 
a sequence of trading times such that the orders are consecutively composed by buy and sell trades, for all initial inventories $X_1$ and $X_2$. In the optimal execution literature such behavior is termed {\it transaction triggered price manipulation}, see \cite{Alfonsi}. 
Figure \ref{fig_price} shows the simulation of the price process under the Schied and Zhang model when both investors have an inventory equal to $1$ for two values of $\theta$. The unaffected price process is a simple random walk with zero drift and constant volatility and the trading of the two agents, according to the Nash equilibrium, modifies the price path. For small $\theta$ (top panel) the affected price process exhibits wild oscillations, while when $\theta$ is large (bottom panel) the irregular behavior disappears\footnote{Moreover, we observe that the presence of spurious oscillations in the price
dynamics may affect the consistency of the spot volatility estimation. Indeed, these oscillations
act as a market microstructure noise, even if this noise is caused by the oscillations of a deterministic trend, while usually it is characterized by some additive noise term.
In particular, we find that when $\theta$ is close to zero the noise is amplified by spurious oscillations, while for sufficiently large $\theta$ these oscillations do not compromise the consistency of the spot volatility.}.
        \begin{figure}[!h]
\centering
{\includegraphics[width=0.85\textwidth]{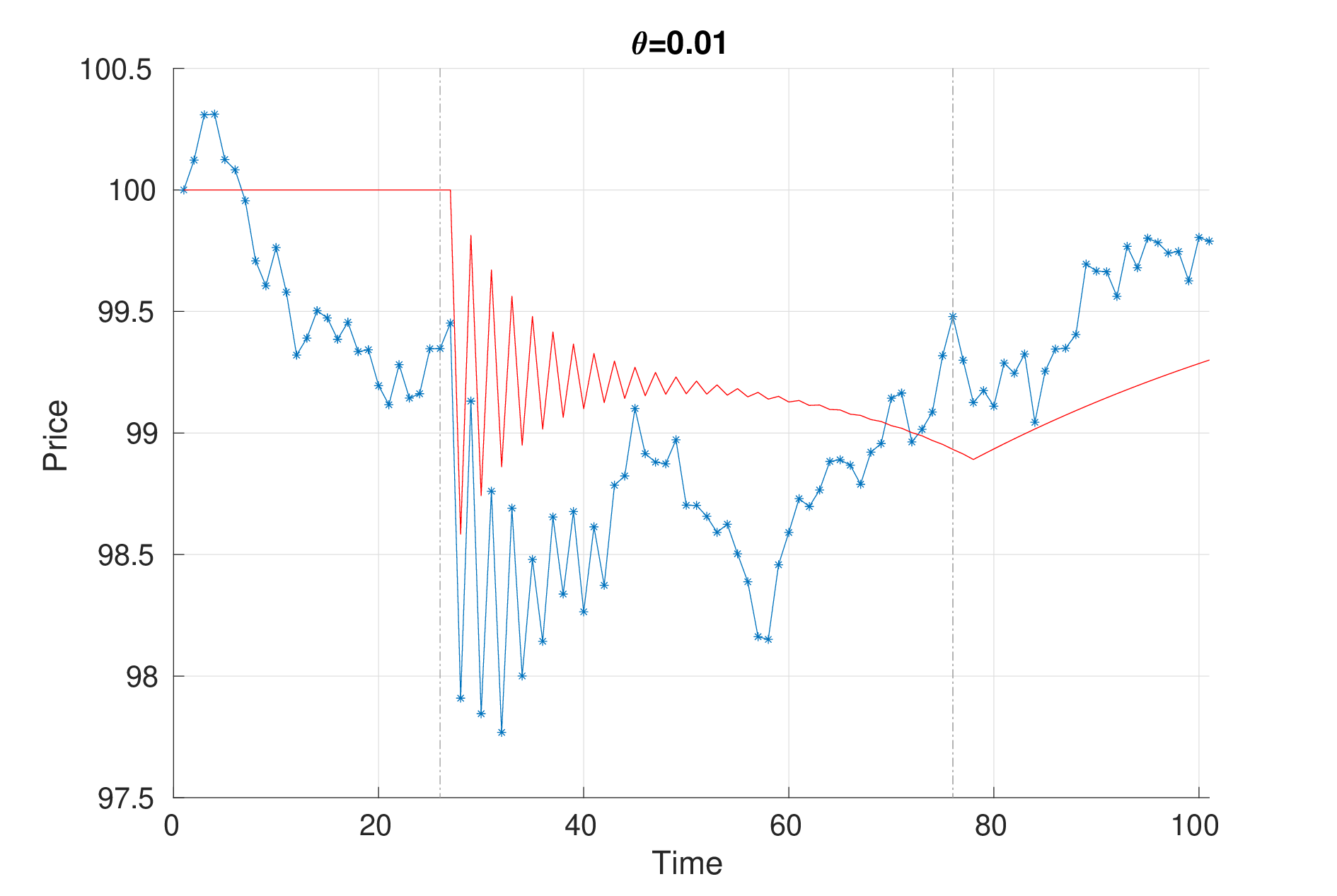}}
{\includegraphics[width=0.85\textwidth]{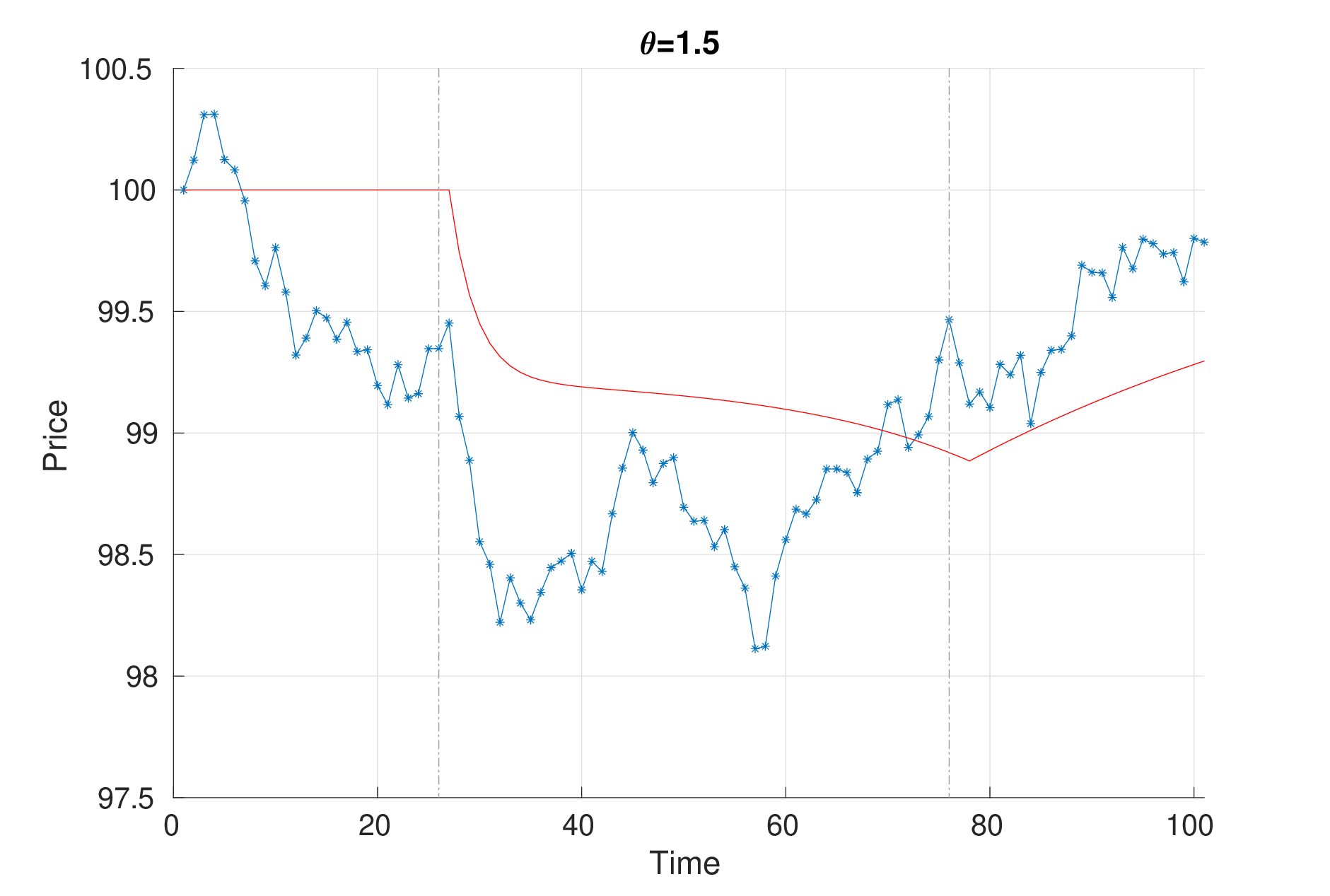}}
\caption{Blue lines exhibit the price process when both agents have inventory equals to $1$. The top (bottom) panel shows the dynamics when $\theta=0.01$ ($\theta=1.5$). The trading time 
grid has $N+1=51$ points, $G(t)=\exp(-t)$, the volatility of the 
unaffected price process is fixed to $1$ and $S_0=100$. The 
vertical grey dotted lines delineates the trading session. The 
red lines shows the drift dynamics due to trading.}\label{fig_price}
\end{figure}

Thus, 
\cite{schied2018market} showed,
when the trading time grid is equispaced, $\numberset{T}_N$,
and under general assumptions on $G$, 
 the existence of a critical value $\theta^*=G(0)/4$
 such that for $\theta<\theta^*$ the 
equilibrium strategies exhibit oscillations of buy and sell orders for both traders. 
Hence, the behavior at zero of the 
kernel function plays a relevant role for the 
equilibrium stability.
Now, we recall the extension of this framework in a multi-agent market ($J >2$) of
\cite{luo_schied}. Then, we first extend 
their framework in the multi-asset ($M>1$) case, where we show the existence and uniqueness of
the related Nash equilibrium.
Finally,
we generalize the stability result of \cite{schied2018market}
in the multi-asset case and we show 
how to appropriately scale with $J$ and $M$ the kernel function in order to prevent huge oscillations in the related equilibria.

\subsection{The Luo and Schied multi-agent market impact model}

The \cite{luo_schied} model 
is an extension of the \cite{schied2018market} model where 
$J$ risk-averse traders want to trade the same asset.
The unaffected price process $S_t^0$ is always assumed to be a right continuous martingale in a suitable filtered probability space
$(\Omega,\F,(\F_t)_{t\geq0},\p)$ and it is also required that $S_{\cdot}^0$
is a square-integrable process. As before, 
let $\T=\{t_0,t_1,\ldots,t_N\}$ be the trading time grid.
 Consistently with the previous notation, we denote with 
 $\Xi=(\xi_{j,k})\in \R^{J\times (N+1)}$
 the matrix of all strategies, where
 $\xi_{j,k}$ is the order flow of agent $j$ at time 
 $t_k$, so that the affected price process is defined as
 \[
 S_t^{\Xi}:=S_t^0 -\sum_{t_k<t} G(t-t_k) 
 \cdot \sum_{j=1}^J \xi_{j,k},
  \]
  where $G$ is the decay kernel.
When comparing markets with a variable number of agents, differently from \cite{luo_schied}, we will assume that the function $G$ can depend on $J$ (\cite{Bagnoli,coimpact}, see below).  The generalization of admissible strategy is straightforward, 
  indeed if $X_{j}$ denotes the inventory of the $j$-th agent, 
  $\Xi$ is admissible for $\bm{X}\in \R^{J}$ and $\T$,
  if $\bm{\xi}_{j,\cdot}$ is admissible for $X_j$ and $\T$
  for each $j$ according to definition \ref{de_admi_strate_schied}, i.e.,
  it is adapted to the filtration, bounded and $\sum_{k=0}^N \xi_{j,k}=X_j$. The set of admissible strategy is denoted 
  as $\mathscr{X}(\bm{X},\T)$.
 Then, 
 if we consider all the possible time priorities among the 
 $J$ traders at each time step, i.e. 
 all the possible permutations that determine the time priority for each trading time $t_k$
 assumed
 to be equiprobable,
 it is possible to generalize the previous definition of liquidation cost for a trader strategy, see \cite{luo_schied} for further details. 
 We denote $\Xi_{-j,\cdot}$ the matrix $\Xi$ where the $j$-th row is eliminated.
\begin{de}[\cite{luo_schied}]\label{de_cost_luo}
 Given a time grid $\T$, the execution costs of a strategy 
 $\bm{\xi}_{j,\cdot}$ given all other strategies  $\bm{\xi}_{l,\cdot}$
 where $l\neq j$ is defined as
 \[
 C_{\T}(\bm{\xi}_{j,\cdot}|\Xi_{-j,\cdot})=
 \sum_{k=0}^{N} \bigg (
 \frac{G(0)}{2}\xi_{j,k}^2- S_{t_k}^{\Xi} \xi_{j,k}+
 \frac{G(0)}{2}\sum_{l\neq j}
 \xi_{j,k}\xi_{l,k}+\theta\ \xi_{j,k}^2
 \bigg),
 \]
 where $\theta\geq0$.
\end{de}

In the framework of \cite{schied2018market}
we have two risk-neutral agents which want to minimize the expected costs of a strategy, i.e. implementation shortfall orders. Now, following \cite{luo_schied}, we consider
the agents' risk aversion by introducing the mean-variance 
and expected utility functionals, respectively
\begin{align} \label{functional_luo_schied}
    MV_{\gamma}(\bm{\xi}_{j,\cdot}|\Xi_{-j,\cdot})&:=
    \E[C_{\T}(\bm{\xi}_{j,\cdot}|\Xi_{-j,\cdot})]+\frac{\gamma}{2}
    \mbox{Var}[C_{\T}(\bm{\xi}_{j,\cdot}|\Xi_{-j,\cdot})],\\
U_{\gamma}(\bm{\xi}_{j,\cdot}|\Xi_{-j,\cdot})&:=
    \E[u_{\gamma}(-C_{\T}(\bm{\xi}_{j,\cdot}|\Xi_{-j,\cdot}))],
\end{align}
where $\gamma$ is the risk-aversion parameter and
$u_{\gamma}(x)$ is the CARA utility function,
\[u_{\gamma}(x)=
\begin{cases}
    \frac{1}{\gamma} (1-e^{-\gamma x})& \quad \text{if } \gamma>0,\\
x& \quad \text{if } \gamma=0.
\end{cases}
\]
As usual, see e.g. \cite{almgren2001optimal}, the minimization of the mean-variance functional is 
restricted to deterministic admissible strategies, which is denoted as 
$\mathscr{X}_{\det}(\bm{X},\T)$.
All agents are assumed to have the same risk-aversion $\gamma\geq0$, 
see \cite{luo_schied} for further details.
Moreover, they introduced the corresponding Nash equilibrium for the 
previously defined functionals.

\begin{de}[from \cite{luo_schied}]
 Given the time grid $\T$ and initial inventories $\bm{X} \in \R^J$
 for $J$ traders with risk aversion parameter $\gamma\neq0$, then: 
\begin{itemize}[leftmargin=0.5in]
    \item  a \emph{Nash Equilibrium for mean-variance optimization} is a matrix of strategies $\Xi^*\in \mathscr{X}_{\det}(\bm{X},\T)$ such that each row $\bm{\xi}^*_{j,\cdot}$ minimizes the mean-variance functional \\ $MV_{\gamma}(\bm{\xi}_{j,\cdot}|\Xi^*_{-j,\cdot})$ over 
    $\bm{\xi}_{j,\cdot}\in\mathscr{X}_{\det}(X_j,\T)$;
    \item 
     a \emph{Nash Equilibrium for CARA expected utility maximization} is a matrix of strategies $\Xi^*\in \mathscr{X}(\bm{X},\T)$ such that each row $\bm{\xi}_{j,\cdot}^*$ maximizes the CARA expected utility functional $ U_{\gamma}(\bm{\xi}_{j,\cdot}|\Xi^*_{-j,\cdot})$ over 
    $\bm{\xi}_{j,\cdot}\in\mathscr{X}(X_j,\T)$.
\end{itemize}

\end{de}

In particular, \cite{luo_schied} showed that 
when the decay kernel is strictly positive definite and
for any $\T$, parameters $\theta,\gamma\geq0$ and initial inventories $\bm{X}\in\R^J$,
there exists a unique Nash equilibrium for the mean-variance optimization which is given by 
\begin{equation}\label{NE_multi_agent}
    \bm{\xi}_{j,\cdot}^* =\overline{X} \bm{v}+(X_j-\overline{X})\bm{w},
    \quad j=1,2,\ldots,J,
\end{equation}
    where $\overline{X}=\frac{1}{J}\sum_{j=1}^J X_j$
    and $\bm{v}$, $\bm{w}$ are the 
    fundamental solutions  defined as
\begin{align*}
   \bm{v}&=\frac{1}{\bm{e}^T [\Gamma^{\gamma,\theta}+(J-1)\widetilde{\Gamma}]^{-1}\bm{e}}
   [\Gamma^{\gamma,\theta}+(J-1)\widetilde{\Gamma}]^{-1}\bm{e}\\ 
     \bm{w}&=\frac{1}{\bm{e}^T [\Gamma^{\gamma,\theta}-\widetilde{\Gamma}]^{-1}\bm{e}}
     [\Gamma^{\gamma,\theta}-\widetilde{\Gamma}]^{-1}\bm{e},
\end{align*}
and, if $\varphi(t):= \mbox{Var}(S_t^0)$, for $t\geq0$, the matrix $\Gamma^{\gamma,\theta}$ is defined for $\theta,\gamma \geq0$ as 
$$\Gamma_{i,j}^{\gamma,\theta}:=(\Gamma_{\theta})_{i,j}+\gamma \varphi(t_{i-1} \wedge t_{j-1}), \quad  i,j=1,2,\ldots,N+1,$$ 
where $\Gamma_{\theta}$
is the previously defined kernel matrix.
     Moreover, if 
    $S_{t}^0=S_0+\sigma B_t$, for $t\geq0$, where $S_0,\sigma>0$ are constants 
    and $B_t$ is a standard Brownian motion, i.e., the unaffected price process is a Bachelier model, then \eqref{NE_multi_agent}
    is also a Nash equilibrium for CARA expected utility maximization
   and it is unique if we restrict all trader strategies 
     to be deterministic, see \cite{luo_schied} for further details.

\section{Multi-asset market impact games}
\label{se_MUlti_assets_case}

We now extend the previous framework allowing the $J$ agents
     to trade a portfolio of $M>1$ assets.
Indeed, agents often liquidate portfolio positions, which accounts in trading simultaneously many assets. In general, the optimal execution of a portfolio is different from many individual asset optimal executions, because of (i) correlation in asset prices, (ii) commonality in liquidity across assets (\cite{Chordia}), and (iii) cross-impact effects. In the following we will focus mainly on the third effect, even if disentangling them is a challenging statistical problem and we will discuss its relations with the 
correlation in asset prices which ensure the existence of Nash 
equilibrium.

To proceed, we first extend the notion of admissible strategy to the multi-asset case. 
     A strategy for $J$ traders during the trading time interval $\T$ 
     for $M$ assets is a multidimensional array  $\Xi=(\xi_{i,j,k})\in \R^{M\times J \times (N+1)}$, where 
     ${\xi}_{i,j,k}$ is the strategy for the $j$-th trader in the $i$-th asset at time step $k$.
Straightforwardly, given a fixed time grid $\numberset{T}$
and initial inventory $X\in\R^{M\times J}$, 
     where each column $j$ contains the inventories of trader $j$ for the $M$ assets, a strategy $\Xi$
     of random variables
     is admissible for $X$ if i)  for all time step  $k$,
     $\Xi_{\cdot,\cdot,k}$ is $\F_{t_k}$-measurable and bounded and ii) $\sum_{k=0}^N \bm{\xi}_{\cdot,j,k}=\bm{X}_j \in \R^M$ for each $j$, where 
     $\bm{X}_j$ is the $j$-th column of $X$. 

The second important point is that the trading of one asset modifies also the price of the other asset(s). This effect is termed \emph{cross-impact}. While \emph{self-impact}
may be attributed to a mechanical and induced consequence of the order book, the 
cross-impact may be understood as an effect related to 
mispricing in correlated assets 
which are exploited by arbitrageurs betting on a 
reversion to normality, see \cite{almgren2001optimal} and \cite{schneider2018cross} 
for further details. Cross-impact has been empirically studied recently, see e.g. \cite{mastromatteo2017trading,schneider2018cross} and its role in optimal execution has been highlighted in \cite{Gerry}.

Mathematically cross-impact is modeled by introducing a function  $\mathcal{Q}^{(J,M)}: \R_+ \times \R^M \to \R^M$
 describing how the trading of the $M$ assets affect their prices at a certain future time. Note that in general the cross-impact function might depend on the number of assets $M$ and on the number of agents $J$. Later we will discuss more in detail how this dependence affects market stability. \cite{schneider2018cross} have discussed necessary
    conditions for the absence of price manipulation for multi-asset 
    transient impact models. They have shown that the cross-impact function need to be symmetric and linear in order to avoid arbitrage and manipulations.
    Moreover, following example 3.1 of \cite{alfonsi2016multivariate} and
    as empirically observed by \cite{mastromatteo2017trading}, we assume the same 
    temporal dependence of $G$ among the assets.
    Then, we assume that  $\mathcal{Q}^{(J,M)}= Q \cdot G(t)$ where $Q$ is linear and symmetric, i.e., $Q\in \R^{M\times M}$ and $Q=Q^T$ and $G:\R_+\to \R_+$. Clearly the dependence from $J$ and $M$ can be in $Q$ and/or in $G(t)$.
    We also assume that $Q$  is a nonsingular matrix. Therefore, the 
   price process during order execution is defined as
 \[
        \bm{S}_t^{\Xi}=\bm{S}_t^0-
    \sum_{t_k<t} 
    G(t-t_k)\cdot Q\cdot  \sum_{j=1}^J \bm{\xi}_{\cdot,j,k} 
    \]
    where we refer to $Q\in \R^{M\times M}$ as the cross-impact matrix, $\bm{S}_t^0\in \R^M$ is 
    the unaffected price process which is 
     assumed to be a right-continuous martingale defined on a suitable 
     filtered probability space and it is a square-integrable process.
 
 If for each asset the time priority among the traders 
 is determined by considering all the possible 
 permutations of agents for each trading time $t_k$,
 then, following the same motivation of \cite{schied2018market} and \cite{luo_schied}, the
  definition  \ref{de_cost_luo} of
  liquidation cost 
  is generalized as follows:
 
   \begin{de}[Execution Cost]
 Given a time grid $\T$ and $\theta\geq0$, the execution cost of a strategy 
 $\Xi_{\cdot,j,\cdot}$ given all other strategies  $\Xi_{\cdot,l,\cdot}$
 where $l\neq j$ is defined as
 \[
 \begin{split}
 C_{\T}(\Xi_{\cdot,j,\cdot}|\Xi_{\cdot,-j,\cdot})&= \sum_{k=0}^{N} \bigg (
 \frac{G(0)}{2}\<Q\bm{\xi}_{\cdot,j,k},\bm{\xi}_{\cdot,j,k}\rangle- 
 \<\bm{S}_{t_k}^{\Xi}, \bm{\xi}_{\cdot,j,k}\rangle+\\&+
 \frac{G(0)}{2}\sum_{l\neq j}
 \<Q\bm{\xi}_{\cdot,l,k},\bm{\xi}_{\cdot,j,k}\rangle+\theta\ \< \bm{\xi}_{\cdot,j,k},\bm{\xi}_{\cdot,j,k}    \rangle.
 \bigg).
 \end{split}
 \]
\end{de}
 The previous definition is motivated by the following 
 argument. When only agent $j$ trades, the prices are moved from $\bm{S}_{t_k}^{\Xi}$ to 
    $\bm{S}_{t_k +}^{\Xi}=\bm{S}_{t_k}^{\Xi} -G(0) Q \bm{\xi}_{\cdot,j,k}$.
    However, the order is executed at the average price and
    the player incurs in the 
    expenses
    \[
     -\frac{1}{2} \langle(\bm{S}_{t_k}^{\Xi}+\bm{S}_{t_k +}^{\Xi}), \bm{\xi}_{\cdot,j,k} \rangle=
     \frac{G(0)}{2} \langle Q \bm{\xi}_{\cdot,j,k},\bm{\xi}_{\cdot,j,k}\rangle- \langle \bm{S}_{t_k }^{\Xi} ,\bm{\xi}_{\cdot,j,k}\rangle.
    \]
    Then, suppose that immediately after $j$ the agent $l$
    place an order and the prices are moved 
    linearly from $\bm{S}_{t_k +}^{\Xi}$ to 
    $\bm{S}_{t_k +}^{\Xi} -G(0) Q \bm{\xi}_{\cdot,l,k}$, so the cost for $l$ is 
    given by:
    \[
     -\frac{1}{2} \langle(\bm{S}_{t_k+}^{\Xi}+\bm{S}_{t_k +}^{\Xi})-G(0)Q\bm{\xi}_{\cdot,l,k}, 
     \bm{\xi}_{\cdot,l,k} \rangle=
     \frac{G(0)}{2} \langle Q\bm{\xi}_{\cdot,l,k},\bm{\xi}_{\cdot,l,k}\rangle- \langle \bm{S}_{t_k }^{\Xi} ,\bm{\xi}_{\cdot,l,k}\rangle
     +G(0) \langle Q\bm{\xi}_{\cdot,j,k} , \bm{\xi}_{\cdot,l,k}\rangle.
    \]

    The term $G(0) \langle Q\bm{\xi}_{\cdot,j,k} , \bm{\xi}_{\cdot,l,k}\rangle$ is the additional cost due to the latency, where on average for each asset 
    half of the times the order of agent $j$ will be executed before the one of agent $l$, so that the latency costs for agent $j$ at time step $k$ is given by 
    $ \frac{G(0)}{2}\sum_{l\neq j}
 \<Q\bm{\xi}_{\cdot,l,k},\bm{\xi}_{\cdot,j,k}\rangle$, 
 see \cite{luo_schied} for further details.

The mean-variance and CARA expected utility functionals are straightforwardly
generalized using the previous defined execution cost.  Indeed,
\begin{align}
    MV_{\gamma}(\Xi_{\cdot,j,\cdot}|\Xi_{\cdot,-j,\cdot})&:=
    \E[C_{\T}(\Xi_{\cdot,j,\cdot}|\Xi_{\cdot,-j,\cdot})]+\frac{\gamma}{2}
    \mbox{Var}[C_{\T}(\Xi_{\cdot,j,\cdot}|\Xi_{\cdot,-j,\cdot})],\\
U_{\gamma}(\Xi_{\cdot,j,\cdot}|\Xi_{\cdot,-j,\cdot})&:=
    \E[u_{\gamma}(-C_{\T}(\Xi_{\cdot,j,\cdot}|\Xi_{\cdot,-j,\cdot}))].
\end{align}
Therefore, we may define
the related Nash equilibria definitions:

\begin{de}
 Given the time grid $\T$ and initial inventories $X\in \R^{M\times J}$
 for $M$ assets and $J$ traders with risk aversion parameter $\gamma\geq0$, then: 
\begin{itemize}[leftmargin=0.25in]
    \item  a \emph{Nash Equilibrium for mean-variance optimization} is a multidimensional array of strategies $\Xi^*\in \mathscr{X}_{\det}(X,\T)$ such that  $\Xi_{\cdot,j,\cdot}^*$ minimizes the mean-variance functional $MV_{\gamma}(\Xi_{\cdot,j,\cdot}|\Xi^*_{\cdot,-j,\cdot})$ over 
    $\Xi_{\cdot,j,\cdot}\in\mathscr{X}_{\det}(\bm{X}_j,\T)$;
    \item 
     a \emph{Nash Equilibrium for CARA expected utility maximization} is a  multidimensional array of strategies $\Xi^*\in \mathscr{X}(X,\T)$ such that each $\Xi_{\cdot,j,\cdot}^*$ maximizes the CARA expected utility functional $ U_{\gamma}( \Xi_{\cdot,j,\cdot}|\Xi^*_{\cdot,-j,\cdot})$ over 
    $\Xi_{\cdot,j,\cdot}\in\mathscr{X}(\bm{X}_j,\T)$.
\end{itemize}

\end{de}

We recall that $\bm{S}_t^0$ follows a Bachelier model 
  if $\bm{S}_{t}^0=\bm{S}_0+ L\bm{B}_t$ where $\bm{S}_0$ is a fixed vector and $\bm{B}_t$ is a multivariate (standard) Brownian motion, where its components are independent with unit variance so that the variance-covariance matrix of $\bm{S}_t^0$ is given by $\Sigma=LL^T$.

    \subsection{Nash equilibrium for the linear cross impact model}\label{solution_linear_cross}
      We now prove the existence and uniqueness of the Nash equilibrium in this multi-asset setting. 
      This is achieved by using the spectral decomposition of $Q$ to
 orthogonalize the assets, which we call ``virtual'' assets, so that 
 the impact of the orthogonalized strategies on the virtual assets 
 is fully characterized by the self-impact, i.e., the transformed cross impact matrix is diagonal.
Thus, the existence and uniqueness of the Nash equilibrium derives immediately by following the same argument as in \cite{schied2018market} and \cite{luo_schied}.
 All the proofs are given in Appendix \ref{app_1}. 
    
    \begin{os}\label{multi_orto}
   If we suppose that $Q$ is the identity matrix, then the multi-asset
    market impact game is 
    a straightforward generalization of the \cite{luo_schied} model. Indeed, each order of the players for the $i$-th stock does not affect 
    any other asset. 
    \end{os}
     
       In general, if we assume that $\bm{S}_t^0$ has uncorrelated components, i.e., 
the variance-covariance matrix $\Sigma$ is diagonal, then
the following result holds.
    
\begin{lem}[Nash Equilibrium for Diagonal Cross-Impact Matrix]\label{lem_NE_diagonal_multi_agent}
If $\bm{S}_t^0$ has uncorrelated components,
for any strictly positive definite decay kernel $G$, time 
grid $\T$, parameters $\theta,\gamma\geq0$, initial
inventory $X\in\R^{M\times J}$ and 
diagonal positive cross impact 
matrix $D=\mbox{diag}(\lambda_1,\lambda_2,\ldots,
\lambda_M)$, there exists a unique 
Nash Equilibrium $\Xi^*\in 
  \mathscr{X}_{\det}(X,\T)$ for the mean-variance optimization problem and it is given by 
  \begin{equation}\label{eq_NE_multi_agent_asset_diagonal}
       \bm{\xi}_{i,j,\cdot}^{*} =\overline{X }_{i,\cdot} \bm{v}_i+(X_{i,j}-\overline{X}_{i,\cdot} )\bm{w}_i, 
      \quad j=1,2,\ldots,J, \quad i=1,2,\ldots,M,
  \end{equation}
    where 
    $\overline{X}_{i,\cdot}=\frac{1}{J}
    \sum_{j=1}^J X_{i,j}$, 
    $\bm{v}_i$ and $\bm{w}_i$ are 
  the fundamental solutions associated with the
  decay kernel $G_i(t)=G(t)\cdot \lambda_i$
  and same parameter $\theta$. Moreover, if
  $\bm{S}^0_t$ follows a Bachelier model,  then \eqref{eq_NE_multi_agent_asset_diagonal}
    is also a Nash equilibrium for CARA expected utility maximization.
  
\end{lem}

    \begin{os}\label{os_risk_neutral_agents}
We observe that for risk-neutral agents, i.e., $\gamma=0$, the assumptions of uncorrelated assets is no more necessary
to prove Lemma \ref{lem_NE_diagonal_multi_agent}. 
Indeed, the mean-variance functional is restricted only to
the expected cost and for linearity $MV_0(\Xi_{\cdot,j,\cdot}|\Xi_{\cdot,-j,\cdot})=\sum_{i=1}^M 
MV_0(\bm{\xi}_{i,j,\cdot}|\Xi_{i,-j,\cdot}; G_i)$, 
where $MV_0(\bm{\xi}_{i,j,\cdot}|\Xi_{i,-j,\cdot}; G_i)=
\E[C_{\T}(\bm{\xi}_{i,j,\cdot}|\Xi_{i,-j,\cdot}; G_i)]$ is the expected cost of Definition \ref{de_cost_luo}
where the decay kernel is multiplied by $\lambda_i$,
and we have the same conclusion of Lemma \ref{lem_NE_diagonal_multi_agent} regardless the covariance matrix of $\bm{S}_t^0$.
\end{os} 
    
    We first introduce some notation and then we state the main
    results. We say that assets are orthogonal if the 
    corresponding cross-impact matrix is diagonal.
    Let us consider the spectral decomposition of $Q$, 
    i.e., $QV=VD$, where $V$ and $D$ are the orthogonal and diagonal matrices containing the eigenvectors and eigenvalues, respectively.
    Since we assume that $Q$ is a non singular symmetric matrix, then 
    $D$ is diagonal with all elements different from zero. 
     We define the prices of
  the virtual assets as $\bm{P}_t:= V^T \bm{S}_t^{\Xi}$ and we observe 
  that 
  \begin{equation}\label{eq_virtual_assets}
  \begin{split}
     \bm{P}_t &=\bm{P}_t^0-\sum_{t_k<t} G(t-t_k) \cdot D \cdot V^T \cdot 
     \bigg(\sum_{j=1}^J \bm{\xi}_{\cdot,j,k}\bigg)
     \\
     &=\bm{P}_t^0-\sum_{t_k<t} G(t-t_k) \cdot D \cdot
     \bigg(\sum_{j=1}^J \bm{\xi}^P_{\cdot,j,k}\bigg),
  \end{split}
  \end{equation}
  where $\bm{P}_t^0:=V^T \bm{S}_t^0$ and 
  $\bm{\xi}_{\cdot,j,k}^P:=V^T \bm{\xi}_{\cdot,j,k}$. This last quantity is 
  the strategy of trader $j$ at time step $k$ in the virtual assets, 
  which is admissible for inventory $\bm{X}_j^P=V^T \bm{X}_j$, i.e, 
  $\sum_{k=0}^N \bm{\xi}_{\cdot,j,k}^P=\sum_{k=0}^N
  V^T \bm{\xi}_{\cdot,j,k}=V^T \bm{X}_j$.
The virtual assets are mutually orthogonal by construction and their 
corresponding (virtual) decay kernels $G_i(t)$ are obtained as the product of the original
decay kernel $G(t)$ and the corresponding eigenvalues $\lambda_i$ 
of the cross impact matrix, i.e., the decay kernel associated with the 
$i$-th virtual asset is $G_i(t):= G(t)\cdot \lambda_i$.
Indeed, from Equation \eqref{eq_virtual_assets} the decay kernel $G(t)$
is multiplied by the eigenvalues of the cross impact matrix for each trading time $t_k$,
  \[
   G(t-t_k) \cdot D=\begin{bmatrix}
                     G(t-t_k) \lambda_1 \\
                      & G(t-t_k) \lambda_2\\
                     & & \ddots \\
                     & & & G(t-t_k) \lambda_M
                    \end{bmatrix}.
  \]
Then, as observed in Remark \ref{multi_orto}, the multi-asset market impact game where each asset is orthogonal to others
is equivalent to $M$ one-asset market impact games, i.e., 
\cite{luo_schied} models. 
The (virtual) decay kernels $G_i(t)$ satisfy the assumptions 
of strictly positive definite kernels as far as $\lambda_i>0 \ \forall i=1,2,\ldots,M$, i.e., 
$Q$ is positive definite
(see also \cite{alfonsi2016multivariate}).
If $\mbox{Cov}(\bm{S}_t^0)=\Sigma$, then $\mbox{Cov}(\bm{P}_t^0)=V^T \Sigma V$. So, if $Q$ and $\Sigma$ are simultaneously diagonalizable
  then $\mbox{Cov}(\bm{P}_t^0)$ is diagonal, i.e., the 
  components of $\bm{P}_{\cdot}^0$ are uncorrelated and 
  by Lemma \ref{lem_NE_diagonal_multi_agent}
we obtain the associated Nash equilibria $\Xi^{*,P}$,
  whose components are defined as
  \begin{equation}
  \label{eq_NASH_multi_asset_agent_virtual}
      \bm{\xi}_{i,j,\cdot}^{*,P} =\overline{X }_{i,\cdot}^P \bm{v}_i+(X_{i,j}^P-\overline{X}_{i,\cdot}^P)\bm{w}_i, 
      \quad j=1,2,\ldots,J, \quad i=1,2,\ldots,M,
  \end{equation}
  where $\overline{X}_{i,\cdot}^P=\frac{1}{J}\sum_{j=1}^J X^P_{i,j}$
  is the average inventory on the $i$-th virtual asset among the traders
   and
  $\bm{v}_i$ and $\bm{w}_i$ are the previously defined fundamental 
  solutions of \cite{luo_schied} for the $i$-th virtual asset $P_{\cdot,i}$. For them, the decay kernel is given by 
  $G_i(t)=G(t) \cdot \lambda_i$ and the corresponding $\varphi_i(t)$
  is given by $\mbox{Var}(P_{t,i}^0)$.
   Since, $Q$ and $\Sigma$ are both symmetric, so diagonalizable, $Q$ and $\Sigma$ are simultaneously diagonalizable if and only if 
 $Q$ and $\Sigma$ commute. Therefore, we consider the following assumption.
\begin{as}\label{ass_cross_impact}
The cross-impact matrix, $Q$, and the covariance matrix of the unaffected 
price process $\bm{S}_t^0$, $\Sigma$, commute, i.e., $Q\Sigma=\Sigma Q.$
\end{as}

This assumption is frequently made in the literature and approximately valid in real data, e.g., \cite{mastromatteo2017trading} makes this assumption
on the correlation matrix.
The empirical observation that the matrix $Q$ has a large eigenvalue with a corresponding eigenvector with almost constant components (as the market factor) and a block structure with blocks corresponding to economic sectors (as in the correlation matrix) indicates that the eigenvectors of $Q$ and $\Sigma$ are the same, i.e. that $Q$ and $\Sigma$ (approximately) commute. Notice also that \cite{garleanu2013dynamic} propose a model of optimal portfolio execution where the quadratic transaction cost is characterized by a matrix which is proportional to $\Sigma$.

We enunciate the following 
theorem of existence and uniqueness of Nash equilibrium which extends Theorem 2.4 of 
  \cite{luo_schied}.
  
  \begin{te}[Nash Equilibrium for Multi-Asset and Multi-Agent Market Impact Games]\label{te_NE_multi_asset_multi_agent}
For any strictly positive definite decay kernel $G$, time 
grid $\T$, parameter $\theta,\gamma\geq0$, initial
inventory $X\in\R^{M\times J}$ and 
symmetric positive definite cross impact 
matrix $Q$ such that Assumption \ref{ass_cross_impact} holds, 
there exists a unique 
Nash Equilibrium $\Xi^*\in 
  \mathscr{X}_{\det}(X,\T)$ for the mean-variance optimization problem and it is given by 
  \begin{equation}\label{eq_NASH_multi_asset_agent}
      \Xi_{\cdot,j,\cdot}^{*}= V\Xi_{\cdot,j,\cdot}^{*,P},
      \quad j=1,2,\ldots,J
  \end{equation}
  where $V$ is the matrix of eigenvectors of $Q$ 
  and $\Xi^{*,P}\in 
  \mathscr{X}_{\det}(X^P,\T)$ is the Nash 
  Equilibrium \eqref{eq_NASH_multi_asset_agent_virtual} of the corresponding
  orthogonalized virtual asset market impact game where 
  $X^P=V^T X$.
  Moreover, if
  $\bm{S}^0$ follows a Bachelier model  then \eqref{eq_NASH_multi_asset_agent}
    is also a Nash equilibrium for CARA expected utility maximization.
\end{te}

  However, we observe that for risk-neutral agents, i.e., $\gamma=0$,
  Assumption \ref{ass_cross_impact} is unnecessary. 
  We remark this result in the following Corollary.
   \begin{co}\label{cor_riks_neutrals}
   If the agents are risk-neutral, i.e., $\gamma=0$, then
for any strictly positive definite decay kernel $G$, time 
grid $\T$, parameter $\theta\geq0$, initial
inventories $X\in\R^{M\times J}$ and 
symmetric positive definite cross impact 
matrix $Q$, 
there exists a unique 
Nash Equilibrium $\Xi^*\in 
  \mathscr{X}_{\det}(X,\T)$ for the mean-variance optimization problem and it is given by 
  \begin{equation}\label{eq_NASH_multi_asset_agent_cor}
      \Xi_{\cdot,j,\cdot}^{*}= V\Xi_{\cdot,j,\cdot}^{*,P},
      \quad j=1,2,\ldots,J
  \end{equation}
  where $V$ is the matrix of eigenvectors of $Q$ 
  and $\Xi^{*,P}\in 
  \mathscr{X}_{\det}(X^P,\T)$ is the Nash 
  Equilibrium associated to the corresponding
  orthogonalized virtual asset market impact game where 
  $X^P=V^T X$ .
  Moreover, if
  $\bm{S}^0_t$ follows a Bachelier model  then \eqref{eq_NASH_multi_asset_agent_cor}
    is also a Nash equilibrium over the set $ \mathscr{X}(X,\T)$.
   \end{co} 


\section{Trading Strategies in Market Impact Games}
\label{cross_trading_conditions}

Before studying market stability we investigate
how the cross-impact effect and the presence of many competitors may affect trading strategies, in terms of Nash equilibria.
To understand the rich phenomenology that can be observed in a market impact game, we introduce three types of traders: 
\begin{itemize}
    \item  the \emph{Fundamentalist} wants to trade one or more assets in the same direction (buy or sell). Notice that a Fundamentalist can have zero initial inventory for some assets; 
    \item the \emph{Arbitrageur} has a zero inventory to trade in each asset and tries to profit from the market impact payed by the other agents;
    \item the \emph{Market Neutral} has a non zero volume to trade in each asset, but in order to avoid to be exposed to market index fluctuations, the sum of the volume traded in all assets is zero\footnote{Real Market Neutral agents follow signals which are orthogonal to the market factor, thus they typically are short on approximately half of the assets and long on the other half. The sum of trading volume is not exactly equal to zero but each trading volume depends on the $\beta$ of the considered asset with 
    respect to the market factor. In our stylized market setting, we assume that all assets are equivalent with respect to the market factor.}. 
\end{itemize}    
     We remark that an Arbitrageur is a particular case of a 
Market Neutral agent in the limit case when the volume to trade
in each asset is zero.
       Clearly in a single-asset market we have only two types of the previous agents, since a Market Neutral strategy requires at least two assets.

    \subsection{Cross-impact effects and liquidity strategies}

To better understand how  cross-impact
affects optimal liquidation strategies,
we consider the 
case of two risk-neutral agents which can (but not necessarily must) trade $M$ assets.
  We show below that the presence of multiple assets and of cross-impact can affect the trading strategy of an agent interested in liquidating only one asset. In particular, we find, counterintuitively, that it might be convenient for such an agent to trade (with zero inventory) the other asset(s) in order to reduce transaction costs.      
   
   
   We focus on the two-asset case, $M=2$, and we analyse the Nash equilibrium when the kernel function has an exponential decay\footnote{All our numerical experiments are performed with exponential kernel as in (\cite{Obizhaeva}). Schied and Zhang shows that the form of the kernel does not play a key role for stability, given that the conditions given above are satisfied.}, $G(t)=e^{-t}$. The first trader is a 
    Fundamentalist who wants to liquidate the position in the first asset, i.e., $X_{1,1}=1$, while the second agent is an 
    Arbitrageur, i.e., 
    $X_{1,2}=0$.  We set an equidistant trading time grid 
    with $26$ points
    and $\theta=1.5$.
    The second asset is available for trading, but let us consider as a benchmark case when both agents trade only the first asset. This is a standard \cite{schied2018market} game. 
    Figure \ref{fig_round_schied}
    exhibits the Nash Equilibrium for the two players. 
    We observe that the optimal solution for the Fundamentalist 
    is very close to the classical U-shape derived under the Transient Impact Model (TIM)\footnote{Given the initial inventory $X$, the optimal strategy in the standard TIM is $\bm{\xi}=\frac{X}{\bm{e}^T \Gamma_{\theta}^{-1} \bm{e}}\Gamma_{\theta}^{-1}\bm{e},$ see
    for further details \cite{schied2018market}.}, i.e., our model when only one agent is present. However,
    the solution is asymmetric and it is more convenient for the 
    Fundamentalist to trade more in the last period of trading.
    This can be motivated by observing that at equilibrium the Arbitrageur
    places buy order at the end of the trading day, and thus she pushes up the price. Then, the Fundamentalist exploits 
    this impact to liquidate more orders 
    at the end of the trading session. We remark that the Arbitrageur earns at equilibrium, since her expected cost is negative (see the caption).

      \begin{figure}[!t]
\centering
{\includegraphics[width=0.85\textwidth]{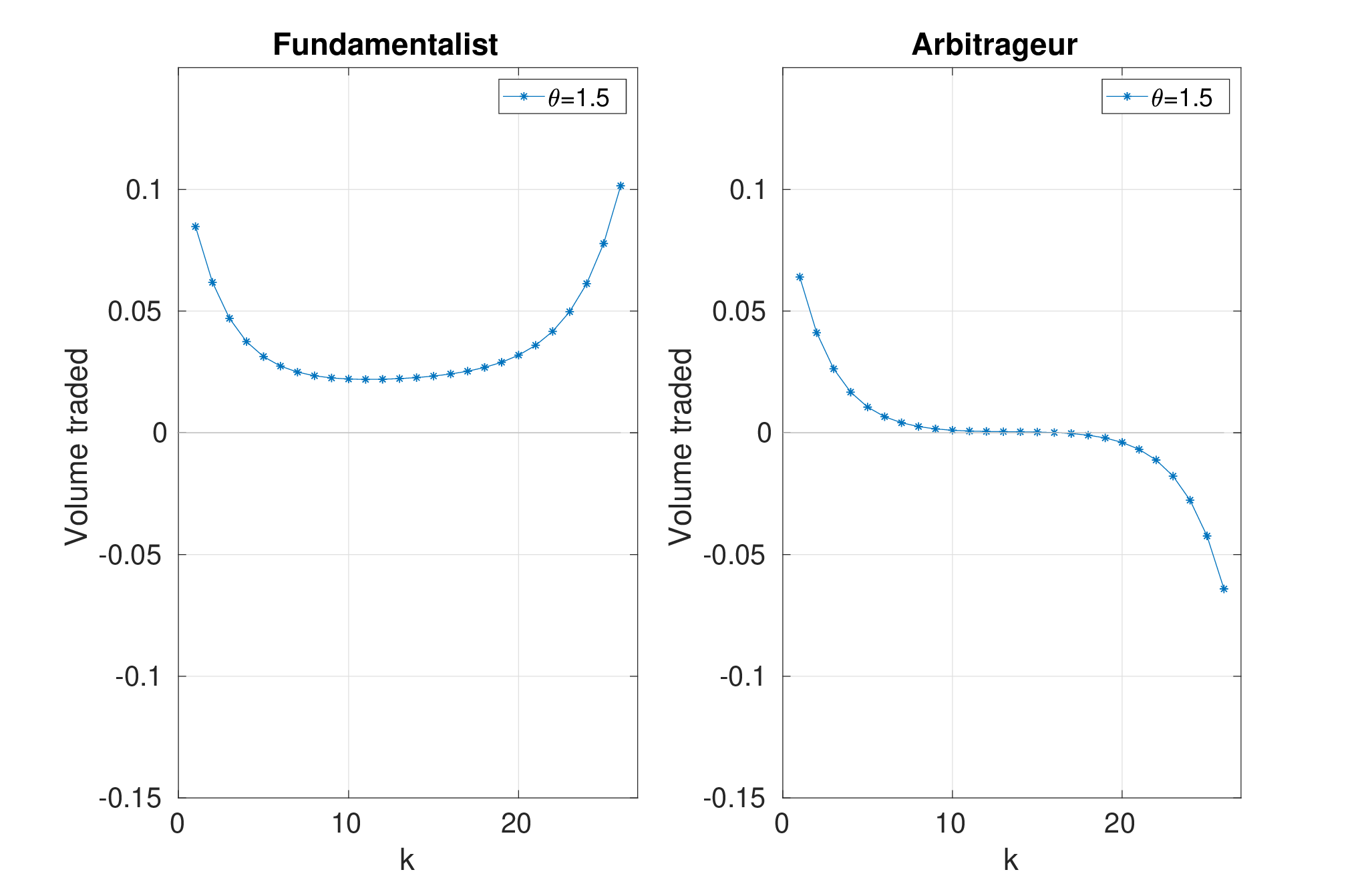}}
\caption{ Nash equilibrium $\bm{\xi}_{1}^*$ of the Fundamentalist
and $\bm{\xi}_{2}^*$ of the Arbitrageur trading only one asset. The trading time grid is equidistant with 26 points and 
$\theta=1.5$. The expected costs are equal to
$\E[C_{\T}(\bm{\xi}_{1}^*|\bm{\xi}_{2}^*)]=0.4882$, 
$\E[C_{\T}(\bm{\xi}_{2}^*|\bm{\xi}_{1}^*)]=-0.0370$.} 
\label{fig_round_schied}
\end{figure}

Now we examine the previous situation when the two traders 
solve the optimal execution problem taking into account the possibility
of trading the other asset. 
We define the cross impact matrix $Q=\begin{bmatrix}
1 &q \\ q & 1
\end{bmatrix}$, where $q=0.6$. In Figure \ref{fig_ne_q2}
we report the optimal solution where the
inventory of the agents 
are set to be ${\bm X}_1= \begin{pmatrix}
1&0
\end{pmatrix}^T$ and ${\bm X}_2= \begin{pmatrix}
0&0
\end{pmatrix}^T$. 
The Fundamentalist wants to
liquidate only one asset, but, as clear from the Nash equilibrium, the cross-impact influences the
optimal strategies in such a way that it is optimal for him/her to trade also the other asset. In terms of cost, for the Fundamentalist trading the two assets is worse off than in the benchmark case (see the values of  $\E[C_{\T}(\Xi_{\cdot,1,\cdot}^*|\Xi_{\cdot,2,\cdot}^*)]$ in captions).
However, if the Fundamentalist trades only asset 1 and Arbitrageur trades both assets, the former has a cost of $0.4935$ which is greater than the expected costs associated with Figure \ref{fig_ne_q2}. Thus, the Fundamentalist {\it must} trade the second asset if the Arbitrageur does (or can do it).

For completeness in Table \ref{tab_exp_costs} we compare the expected costs of both 
Fundamentalist and Arbitrageur when the two agents may decide to trade 
i) both assets, i.e., they consider market impact game and cross-impact
effect, or ii) one asset, i.e., they only consider the market impact game.
It is clear that both agents prefer to trade both assets. Actually, the state where both agents trade two assets is the Nash equilibrium of the game where each agent can choose how many assets to trade.
\begin{table}[t]
\centering
  \setlength{\extrarowheight}{2pt}
    \begin{tabular}{cc|c|c|}
      & \multicolumn{1}{c}{} & \multicolumn{2}{c}{Arbitrageur}\\
      & \multicolumn{1}{c}{} & \multicolumn{1}{c}{$1$ Asset}  & \multicolumn{1}{c}{$2$ Asset} \\\cline{3-4}
      \multirow{6}*{
      \rot{\rlap{Fundamentalist}}}  & $1$ Asset & $(0.4882,-0.0370)$ & $(0.4935,-0.0412)$ \\\cline{3-4}
      & $2$ Asset & $(0.4836,-0.0334)$ & {\color{red}${\bm(0.4885,-0.0377)}$} \\\cline{3-4}
    \end{tabular}
    \vspace{15pt}
    \caption{Payoff matrix of expected costs when the Fundamentalist 
 and Arbitrageur inventories are equal 
 to $(1~ 0)^T$ and $(0~ 0)^T$, respectively. We have highlighted in red 
 the Nash Equilibrium associated with this payoff matrix. The payoff in the $i$-th row and $j$-th column correspond to the game when the Fundamentalist and Arbitrageur decide to trade $i$ and $j$ assets, respectively, i.e., 
 the element in the first row and second column is the payoff when the Fundamentalist trades only the first asset while the Arbitrageur trades both assets.}
 \label{tab_exp_costs}
\end{table}

   \begin{figure}[t]
\centering
{\includegraphics[width=0.85\textwidth]{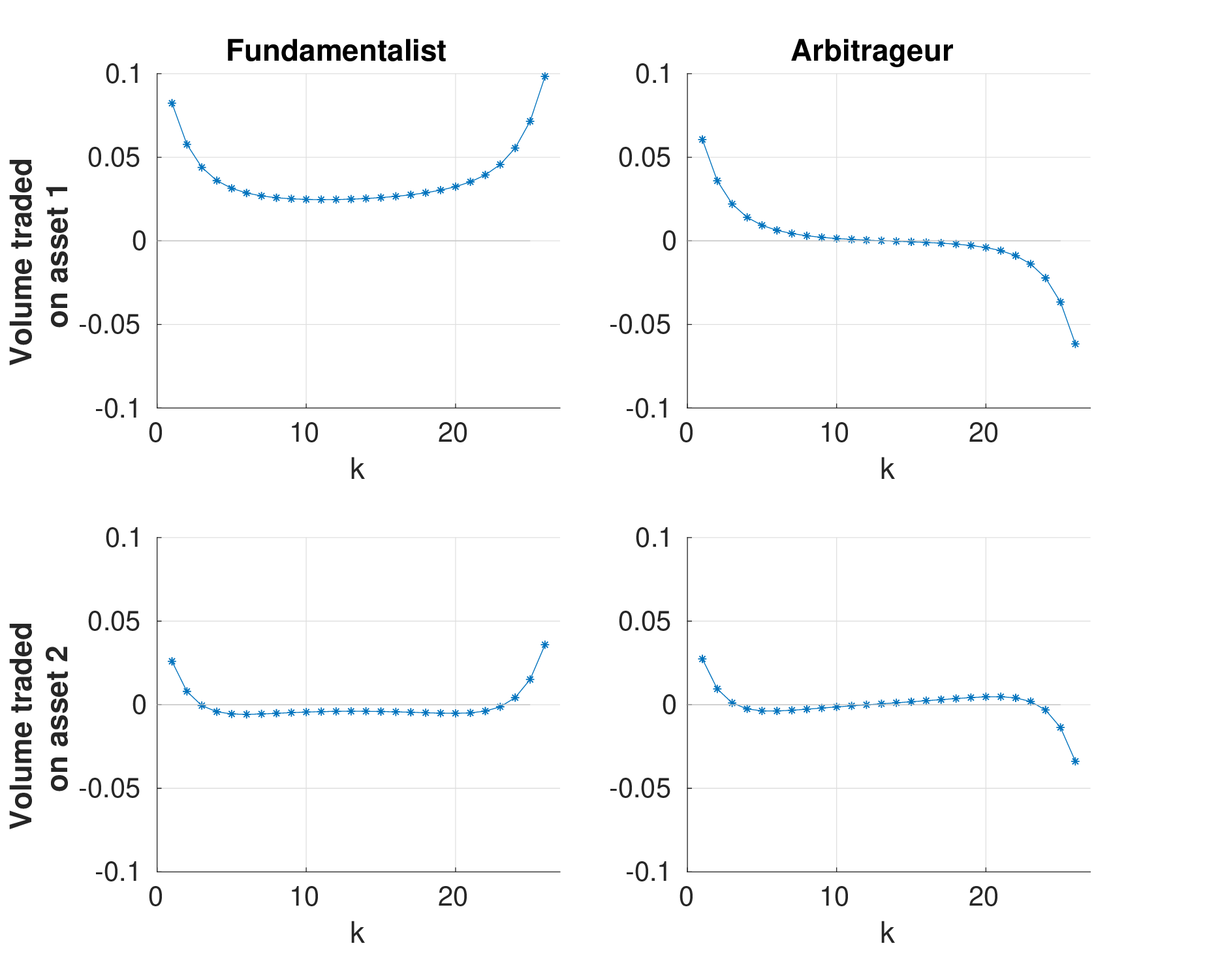}}
\caption{Optimal strategies for a Fundamentalist ($\Xi_{\cdot,1,\cdot}^*$) and an Arbitrageur
($\Xi_{\cdot,2,\cdot}^*$), where their inventories are 
equal to $(1~ 0)^T$ and $(0~ 0)^T$, respectively.
$Q=\protect\begin{bmatrix}
    1 & 0.6 \\ 0.6 & 1    
   \protect\end{bmatrix}$, and the trading time grid is an equidistant time grid with $26$ points. 
 The expected costs are equal
$\E[C_{\T}(\Xi_{\cdot,1,\cdot}^*|\Xi_{\cdot,2,\cdot}^*)]=0.4885$, 
$\E[C_{\T}(\Xi_{\cdot,2,\cdot}^*|\Xi_{\cdot,1,\cdot}^*)]=-0.0377$ when $\theta=1.5$.}
\label{fig_ne_q2}
\end{figure}

The solution presented above is generic, but an important role is played by the transaction cost modeled by the temporary impact. 
When the temporary impact parameter $\theta$ increases,  the benefit of the cross-impact vanishes, and the optimal strategy of the Fundamentalist tends to the solution provided by the simple TIM with one asset and no other agent.
 We find that the difference between these expected costs 
 is negative, i.e. it is always optimal to trade also the second asset, but converges to zero for large $\theta$, see
 Figure \ref{fig_diff_costTIM} panel (a).
Furthermore, it is worth noting that, if $S=\sum_k |\xi_{k,2}|$  denotes the total absolute volume traded by the 
Fundamentalist on the second asset, then 
$\lim_{\theta\to0}S=0$ and $\lim_{\theta\to\infty}S=0$ as
exhibited from Figure \ref{fig_diff_costTIM} panel (b). This means, that
when the cost of trades increases, it is not anymore convenient for 
both traders to try to exploit the 
cross impact effect.

 

  \begin{figure}[t]
\centering
\subfloat[][]
{\includegraphics[width=0.5\textwidth]{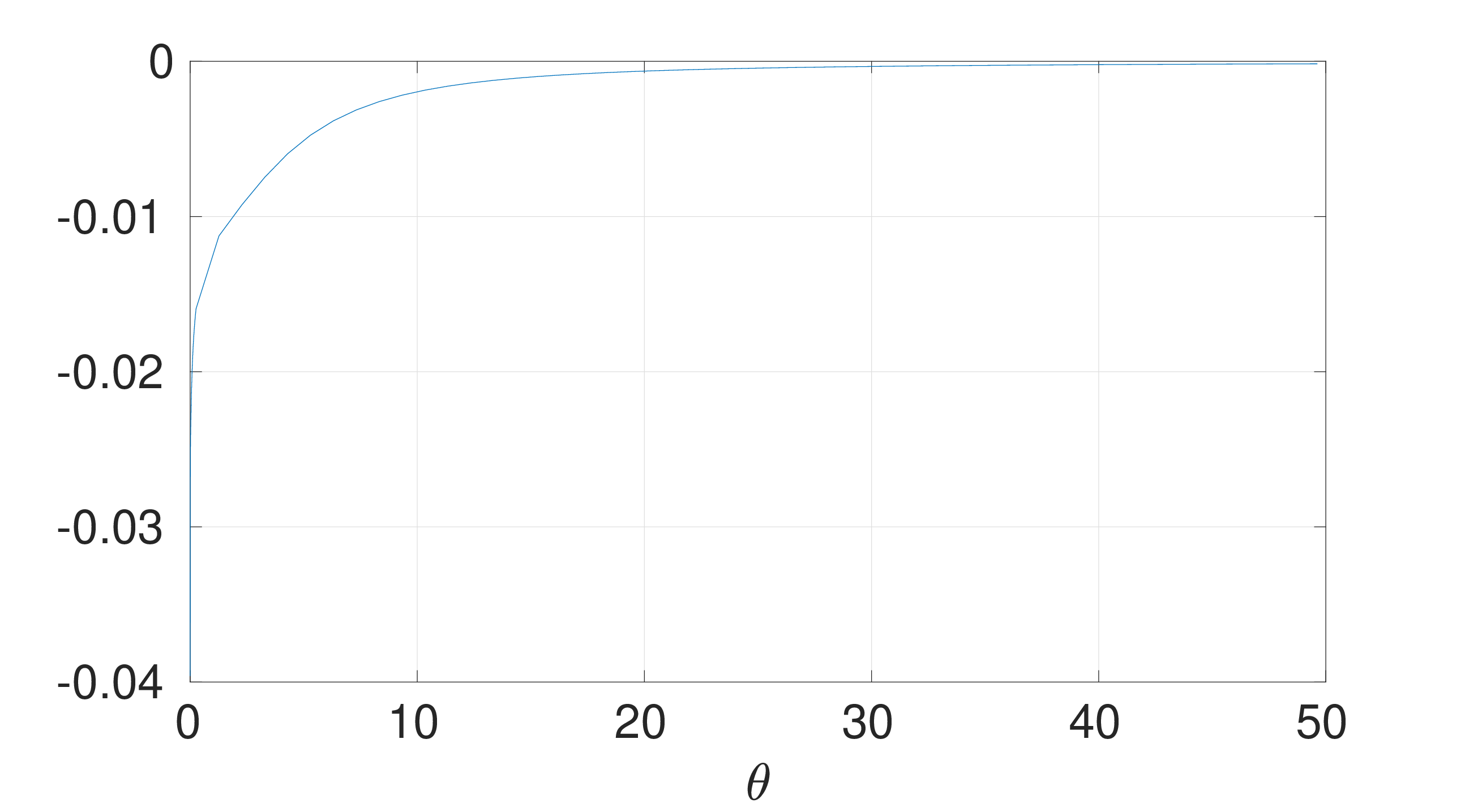}}
\subfloat[][]
{\includegraphics[width=0.5\textwidth]{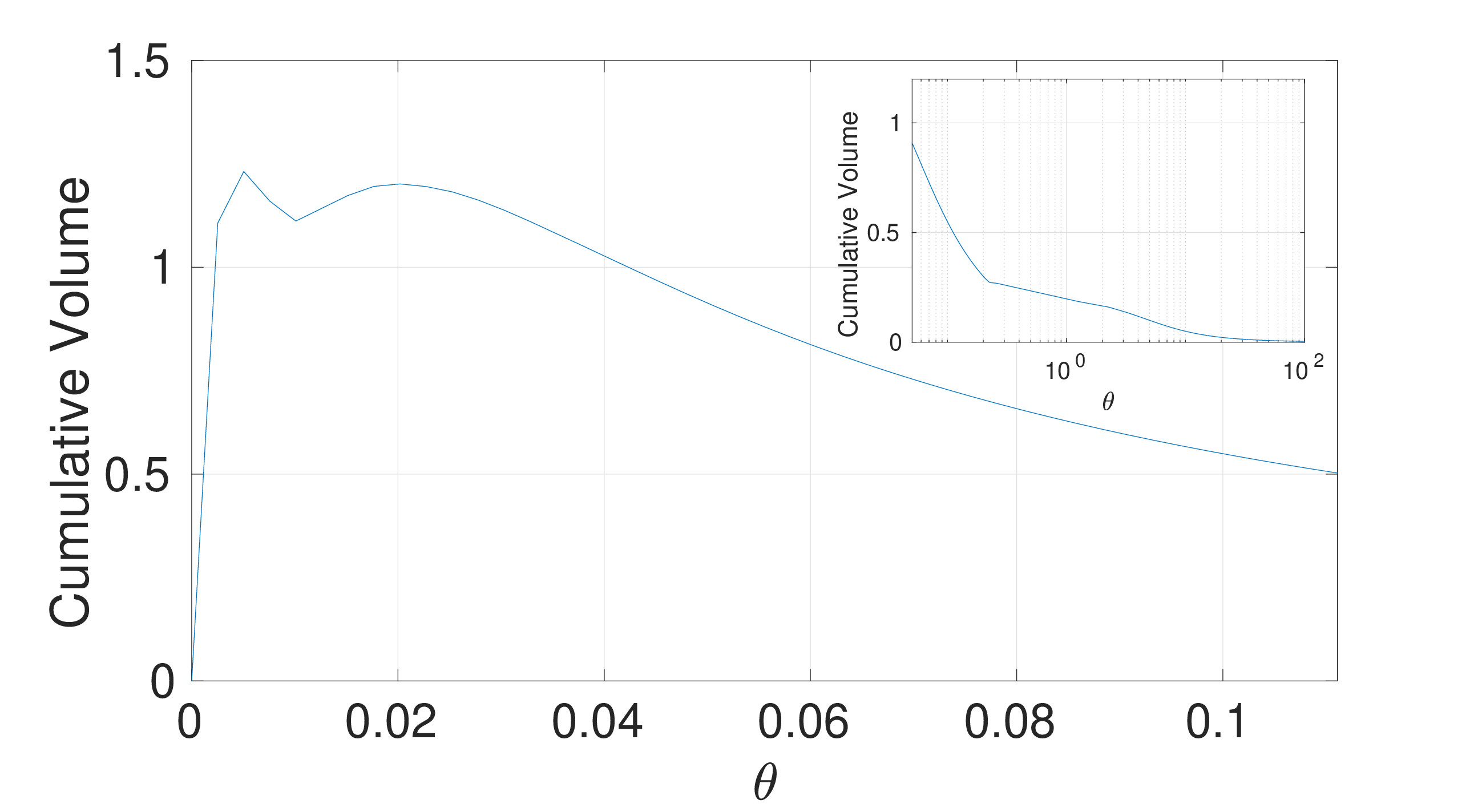}}
\caption{Figure (a). The y axis shows the difference between the expected cost of the Fundamentalist when he/she consider the cross-impact effect and the Arbitrageur and 
the expected cost when he/she places order following the
classical one asset TIM model and the x axis the cost parameter $\theta$. 
Figure (b). Cumulative traded volume of the second asset by the Fundamentalist when playing against an Arbitrageur as a function of $\theta$. The inset shows the same curve in semi-log scale.
The setting is the same of Figure \ref{fig_ne_q2}.
}
\label{fig_diff_costTIM}
\end{figure}

\subsection{Do arbitrageurs act as market makers at equilibrium?}
   
   
 We now consider the cases when the agents are of different type. In particular, we focus on the role of an Arbitrageur
as an intermediary between two Fundamental traders of opposite sign. When a Fundamental seller and a Fundamental buyer trade the same asset(s), are the Arbitrageurs able to profit, acting as a sort of market maker by buying from the former and selling to the latter? 
 
 To answer this question, we  compute the Nash equilibrium of a market impact game with $M=2$ assets and $J=3$ agents, namely a Fundamentalist seller with inventory $(1\ 0)^T$, a Fundamentalist buyer with inventory $(-1\ 0)^T$,
and an Arbitrageur. We assume that agents are risk-neutrals, $\gamma=0$, and 
$Q=\begin{bmatrix}1 & 0.6 \\ 0.6 & 1\end{bmatrix}
$. As panels (a) of Figure \ref{fig_3_agent} show, the Arbitrageur does not longer trade and the expected costs are $0.1056$ and $0$ for the two Fundamentalists and the Arbitrageur, respectively. This indicates that the two Fundamentalists are able to reduce significantly their costs with respect to the previous case, increasing their protection against predatory trading strategies and that the Arbitrageur is unable to act as a market maker.  
The previous cases are particular examples of the following more general result.

\begin{pr}\label{pr_kick_out_arbi}
Under the assumptions of Theorem \ref{te_NE_multi_asset_multi_agent}, 
the following are equivalent:
\begin{itemize}
\item[a)] The aggregate net order flow is zero for each asset, 
i.e., $$\overline{X}_{i,\cdot}=\frac{1}{J}\sum_{j=1}^J X_{i,j}=0\quad \forall i=1,2,\ldots,M;$$
\item[b)] The optimal solution for an Arbitrageur is equal to 
zero for all assets.
\end{itemize} 
\end{pr}

 In other words, when the aggregate net order flow is zero for each asset then there are no arbitrageurs in the market, i.e., 
  the Nash equilibrium for Arbitrageurs
 is \emph{zero}, so that the optimal schedule corresponds to place no orders in the market. 

                \begin{figure}[t]
\centering
{\includegraphics[width=1\textwidth]{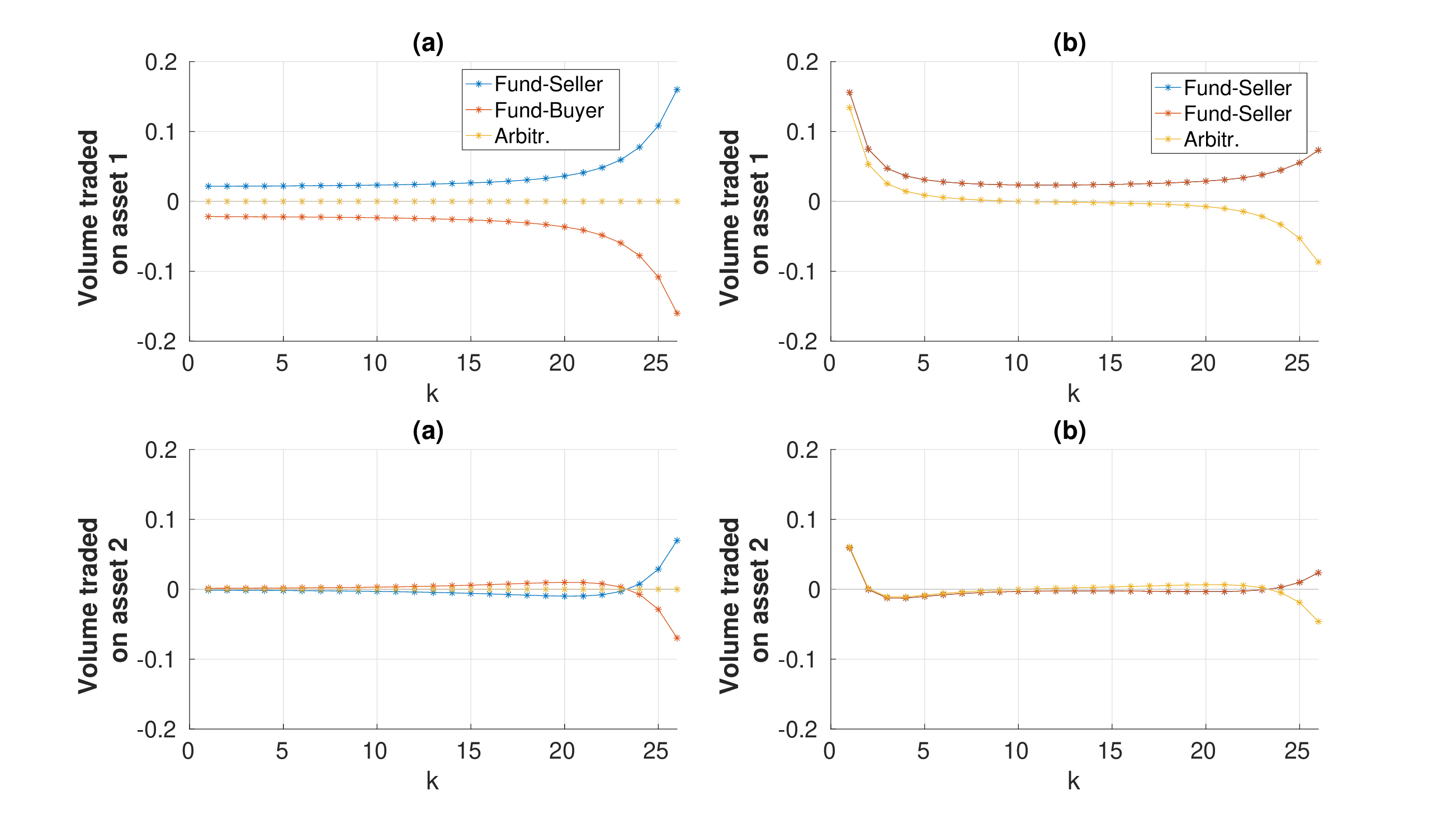}}
\caption{
Optimal schedule for market impact game with $M=2$ assets and $J=3$ risk-neutral agents.
Panels (a) exhibit the optimal schedule for a 
Fundamentalist seller, buyer (with inventory $(1 \ 0)^T$
and $(-1 \ 0)^T$, respectively),  and an Arbitrageur. 
Panels (b) exhibit the optimal schedule for two identical 
Fundamentalist sellers (with inventories $(1 \ 0)^T$, respectively),  and an Arbitrageur. 
Blue and red lines are the Nash equilibrium for the Fundamentalist traders. The yellow line refers to the equilibrium of the Arbitrageur.
The trading time is equidistant with $26$ points, where the cross impact is set to $q=0.6$, $\gamma=0$ and $\theta=1.5.$
}\label{fig_3_agent}
\end{figure} 

As a comparison, we consider two identical Fundamentalist sellers (with inventories $(1\ 0)^T$) and the other parameters are the same as above. Figure \ref{fig_3_agent}, panels (b), displays the equilibrium solution. The solution of the Fundamentalists are identical.
While the trading pattern of the Arbitrageur is qualitatively similar to the one of the two agent case (see Fig. \ref{fig_ne_q2}), the Fundamentalists trade significantly less toward the end of the day. This is likely due to the fact that it might be costly to trade for one Fundamentalist given the presence of the other. The expected costs of the two Fundamentalists is equal to  $0.8911$ (which is approximately two times of the 
two players game) and $-0.0996$ for the Arbitrageur.

\section{Instabilities in Market Impact Games}
\label{sec_inst_comm}
We now turn to our attention to the study of market stability.
Since the seminal work of \cite{schied2018market} we known that, when two risk-neutral agents trade one asset, stability is fully determined by the value $\theta$ of the transaction cost, see Theorem 2.7 of \cite{schied2018market}. Here we extend their results for the multi-asset case and we derive a general result which involves the spectrum of the cross-impact matrix.
However, the proof of  \cite{schied2018market}
  cannot\footnote{The proofs provided of 
 \cite{schied2018market} rely on general results of Toeplitz matrix, which cannot be used in the multi-agent framework, since the involved decay kernel matrices are no longer Toeplitz.} be extended to the multi-agent case with $J$
 risk-averse agents, even though in the one asset case,
 as highlighted by \cite{luo_schied}.
Therefore, we study market stability by using numerical analyses for the general setting of multi-agent and multi-asset case from which we deduce a new conjecture which is in line with the analyses carried out by \cite{luo_schied}. We conclude by presenting some 
advice to policy regulators which want to prevent market instability.

To clarify better our results, we introduce two definitions of market stability in a market 
with $M$ assets and $J$ traders:

\begin{de}[Strong Stability]\label{de_ss}
The market is 
\emph{strongly (uniformly) stable}
if $\forall$ $\theta\geq0$ the
Nash equilibrium 
$\bm{\xi}_{i,j,\cdot}^*\in \mathscr{X}(X_{i,j},\T)$ does 
not exhibit spurious oscillations 
$ \forall\ X_{i,j} \in \R$ initial inventory, for all assets $i=1,2,\ldots,M$ and agents $j=1,2,\ldots,J$.
\end{de}

\begin{de}[Weak Stability]\label{de_ws}
The market is 
\emph{weakly stable}
if there exists an interval $ I \subset \R_+$ such that
$\forall$ $\theta \in I$ the Nash equilibrium
$\bm{\xi}_{i,j,\cdot}^*\in \mathscr{X}(X_{i,j},\T)$ does 
not exhibit spurious oscillations 
$ \forall\ X_{i,j}\in \R $ initial inventory, for all assets $i=1,2,\ldots,M$ and agents $j=1,2,\ldots,J$
\end{de}

We recall that a spurious oscillations is 
a sequence of trading times such that the orders are consecutively composed by buy and sell trades,
see Section \ref{sec_market_impact}. 
Therefore, \cite{schied2018market} showed that for $M=1$ and $J=2$
 the market is not strongly but only 
weakly stable where $I$, the stability region, 
is equal to $[\theta^*,+\infty)$ where 
$\theta^*=G(0)/4.$

\subsection{Scaling of impact with $J$ and $M$}\label{sec:scaling}
Up to now we have not discussed how the function ${\mathcal Q}^{(J,M)}$, and therefore its components $Q$ and $G(t)$, depend on the number of agents $J$ and the number of assets $M$. While this is not important for finding the Nash equilibrium, we will show below that the behavior of ${\mathcal Q}^{(J,M)}$ is critical to study the stability properties of markets. In this subsection we review what the theoretical and empirical literature tells us about this dependence.

Concerning the dependence of impact on $J$,  \cite{Bagnoli} generalizes the \cite{Kyle} model to the case when $J\ge 1$ symmetrically informed agents are simultaneously present, and shows that the Kyle's lambda, i.e. the proportionality factor between price impact and aggregated order flow,  scales as $J^{-1/\alpha}$, where $\alpha$ is the exponent of the stable law describing the price and uninformed order flow distribution. Moreover if the second moment of both variables is finite, \cite{Bagnoli} shows that the Kyle's lambda scales as $1/\sqrt{J}$ (see also \cite{Lambert} for the non symmetrical case when distributions are Gaussian). 
In our impact model, this property can be modeled by assuming that the decay kernel depends on $J$ as
$G(t):=J^{-\beta} \cdot \bar G(t)$ where $\bar G(t)$ is 
 the $J$ independent part of the decay kernel and $\beta\geq 0$.
 The case $\beta=0$ corresponds to the 
 additive case, while for $\beta=1$ the total instantaneous impact does not depend on the number of agents $J$.
 On the empirical side, there are some recent evidences suggesting that the impact strength depends on the number of agents simultaneously trading. Figure 3 of \cite{coimpact} indicates that market impact of a metaorder\footnote{A metaorder is a sequence of trades  executed in the same
direction (either buys or sells) and originating from the
same market participant. Thus in our framework each trader $j$ executes a metaorder of size $X_j$.}  decreases with the number of metaorders simultaneously present.

 Concerning the dependence of impact from $M$, the recent work of \cite{garcia2020multivariate} proposes a multi-asset version of the Kyle model. 
  In particular, they prove the existence and uniqueness of 
 the linear equilibrium and show
 in Proposition 3.4 that the cross-impact matrix $Q$ ($\Lambda$ in their notation) satisfies
 $\frac{1}{4} \Sigma_0=Q \Omega Q$ where $\Sigma_0$ and $\Omega$ 
 are the covariance matrices of the fundamental price and of the bids of 
 the noise trader, respectively.
If we assume that these matrices have a one factor structure and can be decomposed as 
$\Sigma_{0}=s_d I+s_n \bm{e}\bm{e}^T$ and $\Omega=\omega_d I+ \omega_n \bm{e}\bm{e}^T$, 
where $s_d,s_n,\omega_d,\omega_n\in \R$ and $\bm{e}$ is the vector with all components equal to one, then 
necessarily\footnote{If $A=aI+b\bm{e}\bm{e}^T$  is a one-factor matrix, a particular case of 
rank-one update matrix, where $a, b\in \R$, then $A^{-1}$ is one-factor
 and moreover $L$ is 
also a one-factor matrix
where $LL^T=A$. Therefore, by Theorem 3.5 equation (3.6) of \cite{garcia2020multivariate}
$Q$ is a one-factor matrix.} $Q=q_d I+q_n \bm{e}\bm{e}^T$, where $q_d,q_n \in \R$.
 However, since $\frac{1}{4} \Sigma_0=Q \Omega Q$, 
 $q_d=\sqrt{\frac{s_d}{4\omega_d}}$ which is independent from $M$ 
(this was also empirically observed by \cite{benzaquen2017dissecting})
and more interestingly, $\lim_{M\to \infty}q_n=\lim_{M\to \infty} \left(-q_d\pm \sqrt{
\frac{s_d}{4\omega_n}}\right)\frac{1}{M}=0$. \footnote{It is worth to nothing that even though 
$\omega_n=0$, i.e.,  the order flow of the uninformed trader is uncorrelated,
the cross-impact Kyle lambda, $q_n$, is different from zero and $\lim_{M\to \infty}q_n=
\lim_{M\to \infty} 
\sqrt{\frac{s_n}{4\omega_d M}}=0$. } Thus, in the model the off-diagonal terms of $Q$ scale as $1/M$
and asymptotically, when $M$ becomes large, the cross-impact terms would vanish
and $Q$ would
converge towards a diagonal matrix. The decay of cross-impact coefficients with the number of considered assets $M$ has been empirically observed in \cite{benzaquen2017dissecting}.

In conclusion, theoretical and empirical studies have shown that market impact is generally dependent on the number of assets and on the number of agents. While we have some indication of the scaling properties in some specific cases, the general form of this dependence is still an open issue. In the following we will show that if market impact does not properly scale with $J$ and $M$, markets become more unstable when more assets and/or more agents are present.

\subsection{Market stability and cross impact structure }\label{sec_inst_comm_cross_impact}
In this Section we consider $J=2$ risk-neutral agents which trade $M>1$ assets.
We study whether the
increase of the number of assets and the structure of cross impact matrix help avoiding 
oscillations and market instability at equilibrium
according to the previous definitions. To this end, we consider different structures of the cross-impact matrix $Q$ describing the complexity of the market for what concerns commonality in liquidity. 

We first show that instabilities are generically observed also in the multi-asset case and that actually more assets generally make the market less stable if the elements of the cross-impact matrix do not depend on $M$. For simplicity let us consider $M=2$ assets and a game between a Fundamentalist and an Arbitrageur (similar results hold for different combinations of agents).
We choose $G(t)=e^{-t}$, the cross impact matrix equal  to $Q=\begin{bmatrix}
1 & 0.9\\ 0.9 &1
\end{bmatrix}$, and we consider $\theta=0.3$; remember that for the one asset case the market is stable for this value of $\theta$. Figure \ref{fig_insta_2_ass} shows that for this value of $\theta$ the strategies are oscillating and therefore the market is not strongly stable. More surprisingly, the fact that oscillations are observed for $\theta=0.3$ indicates that the transition between the two stability regimes depends on also on the number of assets and that more assets require larger values of $\theta$ to ensure stability. In the following we prove that this is the case and we determine the threshold value. 
 Figure \ref{fig_insta_2_ass} shows also the case  $\theta=0$. Notably, in this case the oscillations in the second asset disappear. This is due to the fact that, since $\Gamma_{0}^1$, ($\Gamma_0^2$), the 
$\Gamma$ matrix associated with the first (second)
virtual asset is equal to $(1+q) \Gamma$, ($(1-q) \Gamma$), 
 the combination of ``fundamental" solutions $\bm{v}$ and 
$\bm{w}$ are the same for the two virtual assets. Thus, at equilibrium the two solutions for the second asset are exactly zero.
   \begin{figure}[t]
\centering
{\includegraphics[width=0.85\textwidth]{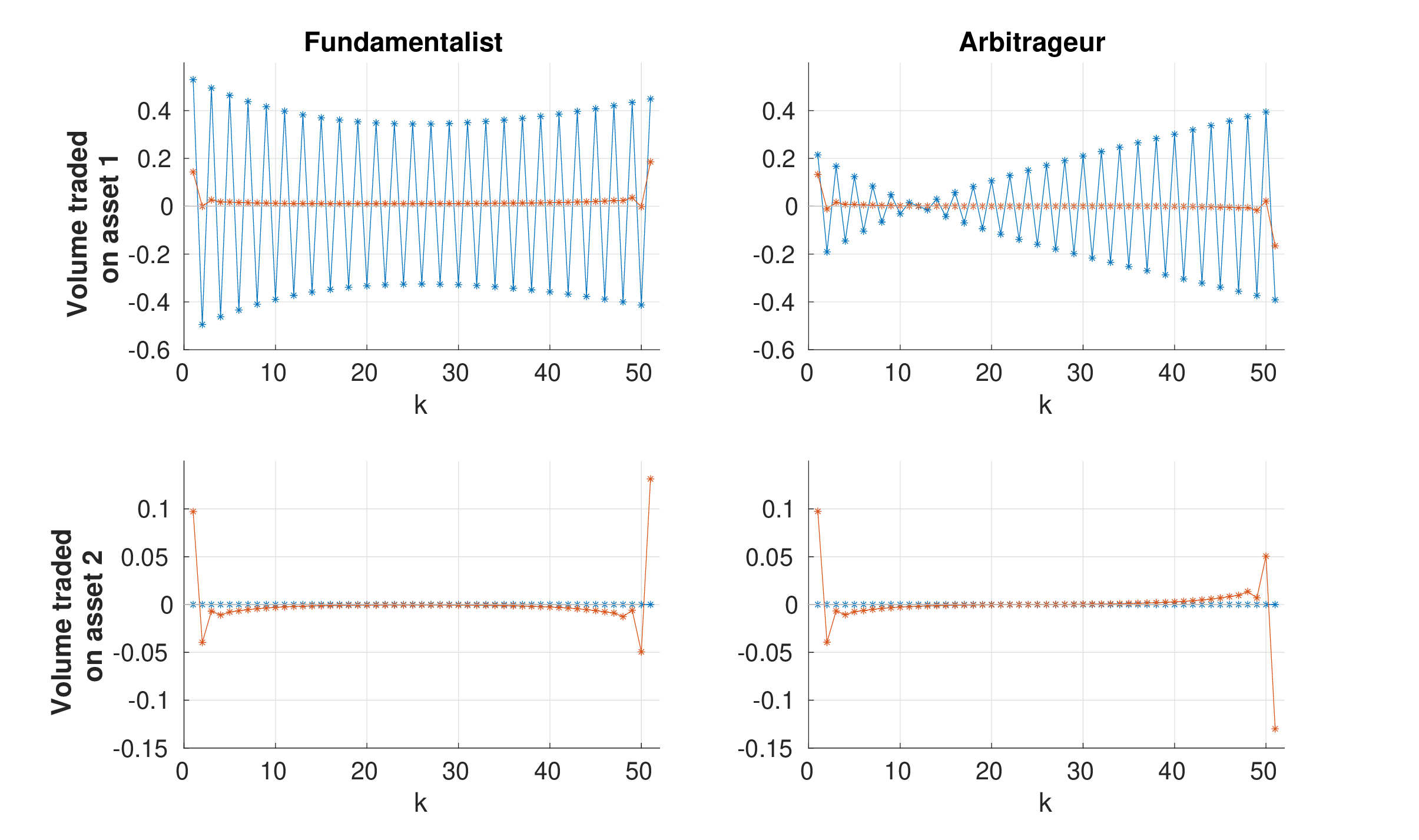}}
\caption{Nash Equilibrium for a Fundamentalist and an Arbitrageur, where 
their inventories are equal to $(1~ 0)^T$ and $(0~ 0)^T$ respectively. The blue lines are the optimal solution when $\theta=0$ and the red lines when $\theta=0.3$. The trading time 
has $51$ points and $Q=\protect\begin{bmatrix}
1 & 0.9\\ 0.9 &1
\protect\end{bmatrix}$.
}
\label{fig_insta_2_ass}
\end{figure}

We have shown in a simple setting that having more than one available asset 
 does not help improving the 
strong stability of the market and increases the threshold value between stable and unstable markets. Now, we show that 
when the number of assets tends to infinity and $G$ does not depend on $M$, the market becomes unstable. 
To this end we introduce the definition of asymptotic stability.

\begin{de}[Asymptotic weak stability]
The market is asymptotically weakly stable if it is 
weakly stable when $M\to \infty.$
\end{de}


Given this definition, we prove the following:

\begin{te}[Instability in Multi-Asset Market Impact Games]\label{te_inst}
 Suppose that G is a continuous, positive definite, strictly positive, log-convex decay kernel and that the time grid is equidistant. Let $({\lambda_i})_{i=1,..,M}$ be the eigenvalues of  the cross-impact matrix $Q$. 
 If $\theta<\theta^*$ the market is unstable, where  
\begin{equation}\label{theta_num}
\theta^*= \max_{i=1,2,\ldots,M} \frac{G(0)\cdot
\lambda_i}{4}.
\end{equation}
 
\end{te}

 Moreover, if the largest eigenvalue of the cross-impact matrix diverges for $M\to \infty$, i.e., $\lim_{M\to +\infty}\lambda_{max}=+\infty$, then 
the market is not asymptotically weakly stable.
The theorem tells that the instability of the market is related 
to the spectral decomposition of the 
cross-impact matrix, i.e. to the liquidity factors. 

We analyze some realistic 
cross-impact matrices and their implications 
for the stability of the Nash equilibrium. 
Specifically, we consider the one-factor and 
block matrices.
\subsubsection{One Factor Matrix}\label{sec_1_fac}
We say that $Q$ is a one 
    factor matrix if
    $Q=(1-q)I+q\cdot  \bm{e}\bm{e}^T$, where $\bm{e}=(1,\ldots,1)^T \in \R^M$ and 
   $q\in (0,1)$ to guarantee the positive definiteness $Q$. As we have seen in Section \ref{sec:scaling}, $q$ can be a function of $M$.
  Then it holds:
   \begin{co} \label{co_1fac}
   Under the assumptions of Theorem \ref{te_inst},
  if the cross-impact matrix is a one factor matrix, then the market is not asymptotically weakly stable if $\lim_{M\to \infty} qM$ diverges.
   \end{co}

   This implies that when $M$ increases and $q$ is independent from $M$, the transactions cost $\theta$ must raise in order to 
   prevent market instability, since $\theta^*=G(0)\lambda_{max}/4\sim G(0)qM/4$, because $\lambda_{max}=1+q(M-1)$. On the contrary, when, as in the multi-assset Kyle model of \cite{garcia2020multivariate}, it is $q=O\left(\frac{1}{M}\right)$, the market is asymptotically stable. Thus the market stability conditions critically depends on the scaling of market impact with $M$.
   
   Figure \ref{fig_fund_arbi_2e2} exhibits 
   the equilibrium for a Fundamentalist and an Arbitrageur, when
   $\theta=1.5$, $q=0.2$ and $M=2000$. The inventory of the Fundamentalist
   is $1$ for the first $1000$ assets and zero for the others.
   The solutions clearly show spurious oscillations of buy and sell 
   orders. Notice that in the one asset case this value of $\theta$ gives a stable market.
          \begin{figure}[t]
\centering
{\includegraphics[width=1\textwidth]{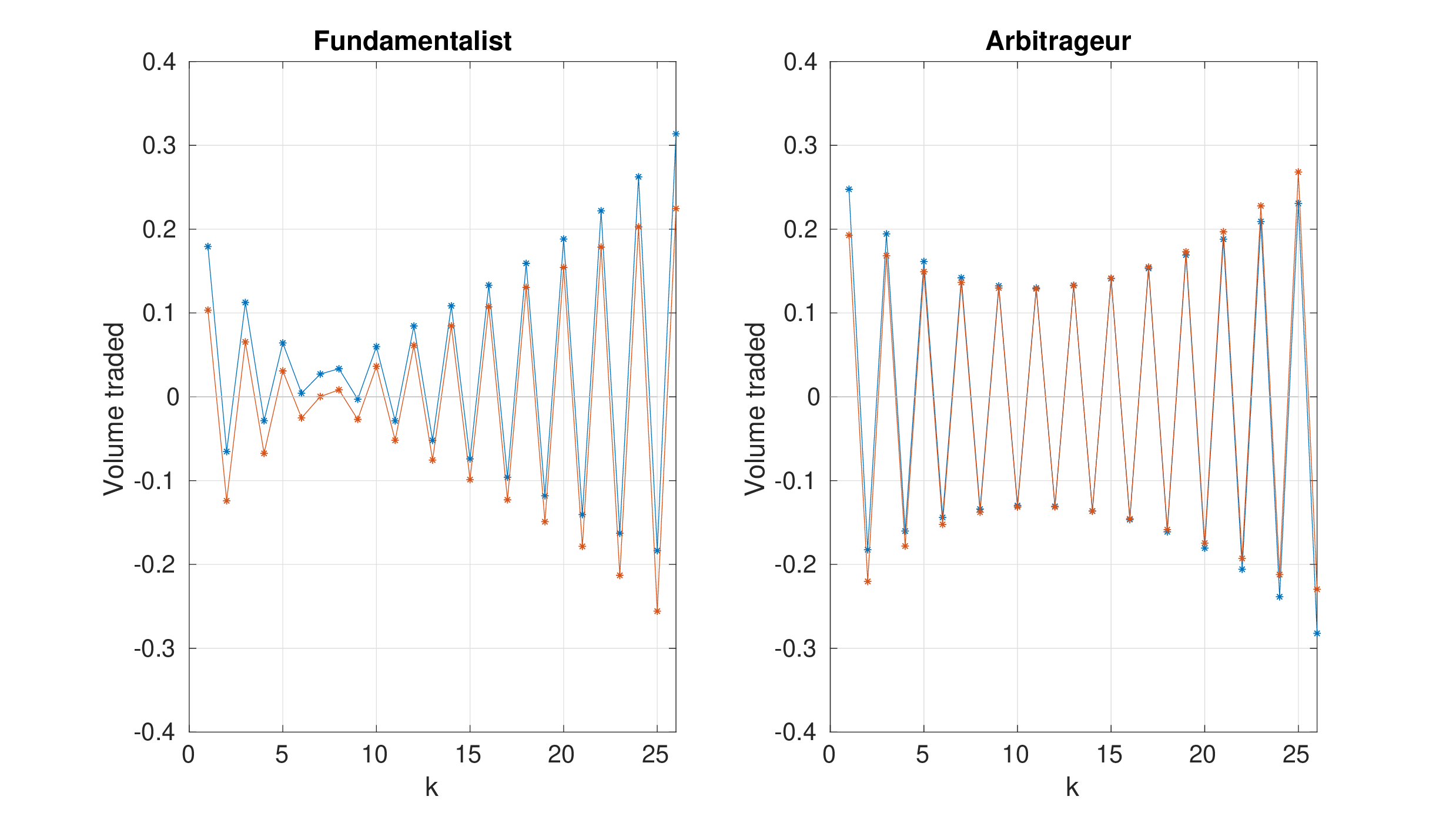}}
\caption{Nash equilibrium when
  $\theta=1.5$ between a Fundamentalist with  inventory
$(1,\ldots,1,0,\ldots,0)^T\in \R^M$ and an Arbitrageur with inventory
$(0,\ldots,0)^T\in \R^M$, where
$M=2,000$. The cross impact matrix is a one factor matrix with $q=0.2$. The blue lines exhibits the volume traded 
for any of the first $1,000$ assets, while the red ones are those for any of the last $1,000$ assets.
The equidistant time grid has 26 points. 
}\label{fig_fund_arbi_2e2}
\end{figure}
   We observe that
   the eigenvector corresponding to $\lambda_{max}$ is given by $\bm{e}$, which 
   represents an equally weighted 
   portfolio. As a consequence, if we consider a Market Neutral agent
   against an Arbitrageur the solution becomes stable $\forall \ \theta > (1-q)/4$, since both traders have zero inventory on the first virtual asset. Thus, oscillations might disappear when the inventory of the agents in the first virtual asset is zero.   

A generalization of the above model considers $Q$ as a rank-one modification matrix, i.e. $Q= D+ \bm{\beta} \bm{\beta}^T$, where $D=\mbox{diag}(1- \beta_1^2 ,\ldots, 1- \beta_M^2)$  and $\bm{\beta}\in \R^M$ is a fixed vector. In this way the cross impact is not the same across all pairs of stocks. 
 We find again that the market is not asymptotically stable if the $\beta$s do not suitably scale with $M$, i.e. instability occurs when $\lim_{M\to \infty} \langle \beta^2_i\rangle M$ diverges, where $\langle ...\rangle$ is the average value over the $M$ values. 

   \subsubsection{Block Matrix}
  We now assume 
    that the cross impact matrix has a block structure in such a way that cross impact between two stocks in the same block $i$ is $q_i$, while when the two stocks are in different blocks the cross impact is $q$,
    which we assume to be $0\leq q<q_i \ \forall i$. This is consistent with the empirical evidence in \cite{mastromatteo2017trading}, where blocks are in good correspondence with economic sectors.
    
   Let us denote with $M_i$ the number of stocks in block $i$, ($i=1,\ldots K$), and
    let $Q_i=(1-q_i)I+q_i\cdot  \bm{e}_i \bm{e}_i^T \in {\mathbb R}^{M_i}\times {\mathbb R}^{M_i}$ with $q_i\in(0,1)$ and $\bm{e}_i=(1,\ldots,1)^T \in \R^{M_i}$, where $K$
   is the number of blocks.
 We define the cross impact matrix as:
   \[
                            Q:=\begin{bmatrix}
      Q_1 &  q \bm{e}_1 \bm{e}_2^T & \cdots & q \bm{e}_1 \bm{e}_K^T\\
      q \bm{e}_2 \bm{e}_1^T &  Q_2 & \cdots & q\bm{e}_2 \bm{e}_K^T\\
      \vdots & & \ddots &  \vdots\\
      q \bm{e}_K \bm{e}_1^T & \cdots & q \bm{e}_K \bm{e}_{K-1}^T &Q_{K}\\
                              \end{bmatrix}.
                           \]
 

 In a similar way of Section \ref{sec_1_fac} $q_i$ might depend on the number of stocks in the $i$-th cluster, $M_i$.
 We prove 
     an analogue result as for the one factor matrix case:

    \begin{co}	\label{cor_block}
    Under the assumptions of Theorem \ref{te_inst},
 if $Q$ is a block matrix, where each block 
is a one factor matrix,
if (i) $\lim_{M_i\to +\infty} M_i (q_i -q) \to +\infty$
for all $i=1,2,\ldots,K$
and (ii) $\lim_{M\to+ \infty}\frac{M}{K}\to +\infty$, then 
the market is not asymptotically weakly stable.
\end{co}
       \begin{figure}[t]
\centering
{\includegraphics[width=1\textwidth]{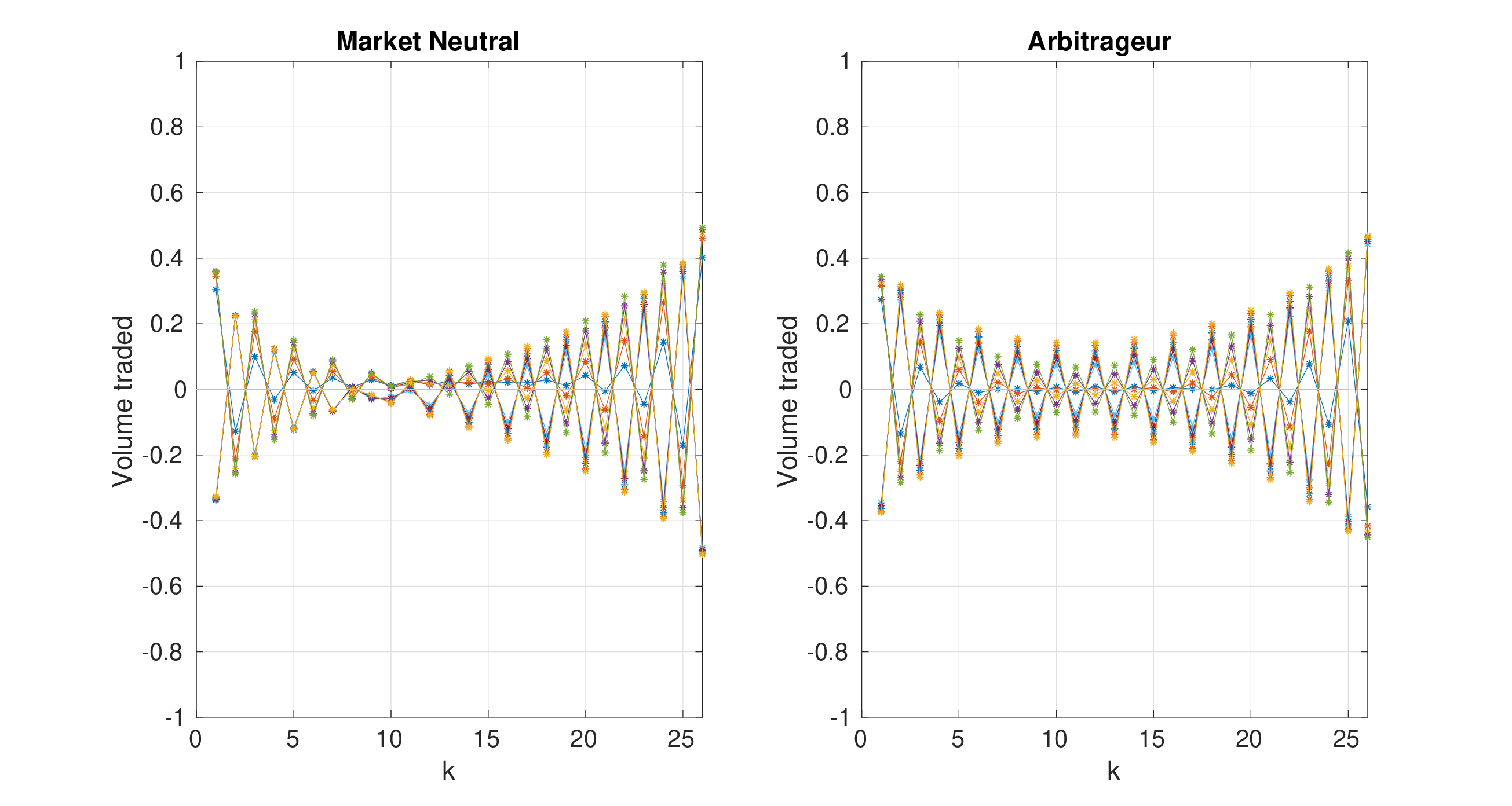}}
\caption{Nash equilibrium when
  $\theta=1.5$ with inventories for the Market Neutral
$X_0=(1,\ldots,1,-1,\ldots,-1)^T\in \R^M$ and
for the Arbitrageur $Y_0=(0,\ldots,0)^T\in \R^M$, where
$M=2000$. The cross impact matrix is a block matrix with $K=10$.
The figure exhibits the equilibria  related to one 
(the first) asset for each block.
 The trading time grid is an 
equidistant time grid with $26$ points. Each block has a cross-impact $q_i$ equal to $0.1, 0.2, \ldots, 0.9$ for $i=1,2,\ldots,9$ and $0.95$ for the last one.}
\label{fig_mn_arb_block_one}
\end{figure}      

When $(q_i-q)$ is independent from $M_i$, for all $i$,
and the average number 
     of stocks of a cluster tends to infinity 
     when $M$ goes to infinity, the transaction 
     costs level parameter must raise in order to prevent instability.
 As an example, we consider $K=10$ equally sized blocks from an universe of $M=2,000$ assets and set $q=0.05$ 
 where each block has a fixed cross-impact $q_i$. 
With this kind of cross impact matrix, we have $K$ large eigenvalues whose eigenvectors correspond to virtual assets displaying oscillations. The optimal trading strategies for stocks belonging to the same block are the same. Thus in Figure \ref{fig_mn_arb_block_one} we show the Nash equilibrium for the first asset in each of the $10$ blocks when the two agents are a Market Neutral and an Arbitrageur. The oscillations are evident, as expected, in all traded assets.



We now study how the critical value $\theta^*$ varies when the number of assets increases for different structures of the cross impact matrix and therefore of the liquidity factors.
Comparing different matrix structures is not straightforward since the critical value depends on the values of the matrix elements. To this end we consider the set of symmetric cross impact matrices of $M$ assets having one on the diagonal and fixed sum of the off diagonal elements. More precisely let $h\in \R$, then we introduce for each $M$ the set 
\[
\mathcal{A}_h^M:=
\{
A\in \R^{M\times M}| A^T =A,  \ \sum_{j=1}^N \sum_{i> j} a_{ij} =h, \  a_{ii}=1
\},
\]
One important element of this set is the cross impact matrix $Q_{1fac}\in {\mathbb R}^{M\times M}$ of a one factor model (see above) with off-diagonal elements equal to $2h/M(M-1)$. In Appendix \ref{app_1} we prove the following: \begin{te}\label{te_one_factor_largest}
For a fixed $h\in \R$, let us consider the related one-factor matrix $Q_{1fac}\in \mathcal{A}_h^M$, then
\[
 \lambda_1(Q)\geq \lambda_1(Q_{1fac}), \quad \forall Q \in \mathcal{A}_h^M,
\]
i.e. among all the matrices with one in the diagonal and constant sum of the off-diagonal terms, the one-factor matrix (i.e. where all the off-diagonal elements are equal) is one of the matrices with the smallest largest eigenvalue. 
\end{te}
Moreover, we prove in the last part of Appendix \ref{app_1} that the previous is not a strict 
inequality, by showing that both a diagonal block matrix, with 
identical blocks, and the one-factor matrix have the same maximum eigenvalue. 
This theorem implies that among all the cross impact matrices belonging to $\mathcal{A}_h^M$, the one factor case is among the most stable cross-impact matrices. For example, it is direct to construct an example of a block diagonal cross impact matrix with non-zero off block elements (i.e. similar to what observed empirically) and to prove that its critical $\theta^*$ is larger than the critical value for the one factor matrix having the same value $h$ of total cross-impact.  

\subsection{Market stability in multi-agent and multi-asset market impact games}
  
  We now study how the stability of the market depends on the number of agents, $J$, the number of assets, $M$, the risk-aversion parameter $\gamma$, and the number of trading periods $N$.
  Specifically, we compute numerically the critical value of $\theta$ after which
   the market is not stable. However, we first 
   observe that to study the stability it is sufficient to analyse the fundamental solutions of each virtual assets.
   
   \subsubsection{Characterization of the fundamental solutions}
\label{sec_fund_solut_chara}
If all agents have the same inventory, i.e.,
$\bm{X}_{\cdot,j}=\bm{Z}\ \forall j$ where $\bm{Z} \in {\mathbb R}^M$ is a fixed inventory vector, then also the virtual inventories are all equal, since
$\bm{X}^P_{\cdot,j}=V^T \bm{Z} \equiv \bm{Z}^P \ \forall j$.
Then, $\overline{X}_{i,\cdot}^P=\frac{1}{J}\sum_{j=1}^J X_{i,j}^P=Z_i^P$ and by Eq. \eqref{eq_NASH_multi_asset_agent_virtual} the solution for all agent $j$ in virtual asset $i$ is given by
$\Xi_{i,j}^{*,P}=Z_{i,j}^P \bm{v}_i.$
So, let $V=\begin{bmatrix}
\bm{\nu}_1 |  \bm{\nu}_2|\cdots
|\bm{\nu}_M
\end{bmatrix}$ the matrix of eigenvectors of $Q$, which 
we may assume to be normalized, $\bm{\nu}_i^T \bm{\nu}_i=1$, 
if $\bm{X}_{\cdot,j}=\bm{\nu}_m \ \forall j$ then
the optimal schedule on the virtual assets is given
\[
\Xi_{i,j}^{*,P}=
\begin{cases}
 \bm{v}_m, & i=m\\
    \bm{0}, & \forall\ i\neq m
\end{cases}, \quad \forall j
\]
since $\bm{X}_{\cdot,j}^P=V^T \bm{X}_{\cdot,j}$ has $1$ in the $m$-th position and zero otherwise, so $
\Xi_{\cdot,j}^{*}=V\cdot \Xi_{\cdot,j}^{*,P}=\bm{\nu}_{m} \otimes \bm{v}_m, \ \forall j,
$
which means that the strategies for all traders 
is fully characterized by the fundamental solution $\bm{v}_m$. 

If $\overline{X}_{i,\cdot}=0, \ \forall i$ then 
$\overline{X}_{i,\cdot}^P=0$ and by equation \eqref{eq_NASH_multi_asset_agent_virtual} the solution for each agent $j$ is given by $\Xi^{*,P}_{i,j}=X_{i,j}^P \bm{w}_i, \ i=1,2,\ldots,M.$. Thus, as for the previous 
case, if the inventory of the $j$-th trader $\bm{X}_{\cdot,j}=\bm{\nu}_m$ (and if 
$\overline{X}_{i,\cdot}=0$ for all $i$), then his/her optimal schedule on the virtual assets is given by
\[
\Xi_{i,j}^{*,P}=
\begin{cases}
 \bm{w}_m, & i=m\\
    \bm{0}, & \forall\ i\neq m
\end{cases},
\] so that $
\Xi_{\cdot,j}^{*}=V\cdot \Xi_{\cdot,j}^{*,P}=\bm{\nu}_{m} \otimes \bm{w}_m,
$.

We summarize the previous results as follows:
\begin{itemize}
\item[a)] If all agents have the same inventories, i.e.
$\bm{X}_{\cdot,j}=\bm{\nu}_m\ \forall j$, then the Nash equilibrium for $j$ is proportional to $\bm{v}_m$, i.e,  $\Xi_{\cdot,j}^{*}=\bm{\nu}_{m} \otimes \bm{v}_m.$

\item[b)] If $\overline{X}_{i,\cdot}=0, \ \forall i$
and $\bm{X}_{\cdot,j}=\bm{\nu}_m$, then the Nash equilibrium 
for $j$ is proportional to $\bm{w}_m$, i.e,  $\Xi_{\cdot,j}^{*}=\bm{\nu}_{m} \otimes \bm{w}_m.$
\end{itemize}

We observe that, respectively, 
if $\bm{v}_m$, or $\bm{w}_m$, exhibits 
spurious oscillations also $\Xi_{\cdot,j}^{*}$ is affected by these oscillations, respectively.
  We recall that market is unstable 
if a particular initial inventories leads to optimal trading
strategies with spurious oscillations.
  So we can restrict the stability analysis on the 
  fundamental solutions among all assets.
  
   \subsubsection{Numerical analysis of stability}
   \label{numerical_analysis_stability}
In this section we study how $J$, $M$, $\gamma$, and $N$ 
affect the market stability 
  in the multi-agent and multi-asset case.
In particular, we compute numerically $\theta^*$
   such that when $\theta<\theta^*$ the market is unstable. As observed in Section \ref{sec_fund_solut_chara} it is sufficient 
   to examine the oscillations of the fundamental solutions on 
   the virtual assets.
   
   We consider the following setting: 
  \begin{itemize}
      \item The time grid is equidistant $\T_N=\{\frac{kT}{N}|k=0,1,\ldots,N\}$, where $T=1$ and $N\in \N$;
      \item The decay kernel is exponential, $G(t)=e^{-t}$;
      \item The cross-impact matrix is a one factor matrix, $Q=(1-q)I_M+q \bm{e} \bm{e}^T$, where $q=1/2$;
      \item $\bm{S}_t^0$ follows a Bachelier model where the covariance matrix is equal to $Q$.
  \end{itemize}
 \cite{luo_schied}  conjectures that in the one-asset case $\theta^*$  satisfies 
 \[\sup_{N,\gamma} \theta^*(1,J,N,\gamma)=G(0)\cdot \frac{J-1}{4},\] 
 therefore, given the results of Section \ref{sec_inst_comm}, our conjecture is that 
 \begin{equation}\label{eq_conjecture}
 \sup_{N,\gamma} \theta^*(M,J,N,\gamma)=G(0)\cdot \frac{(J-1)\lambda_{max}}{4},
 \end{equation} where $\lambda_{max}$ is the maximum eigenvalue of $Q$.
 We recall that in the above setting, $\lambda_{max}=1+\frac{M-1}{2}$ and $G(0)=1$. Moreover we are assuming that $G$ and $Q$ do not depend explicitly on $J$ and $M$.
   Thus, in the first analysis we set $N=300$, $\gamma=10$ and we compute $\theta^*$ as a function of $M$ and $J$.
   Figure \ref{fig_theta*_M_J} exhibits the corresponding level curves.  
   It is worth noticing that the relation between $J$ and $M$ is very close to that of Equation \ref{eq_conjecture}.
   Indeed, the average relative discrepancy on $\theta^*$ is of the order of $10^{-3}$. 
       \begin{figure}[!t]
\centering
{\includegraphics[width=0.75\textwidth]{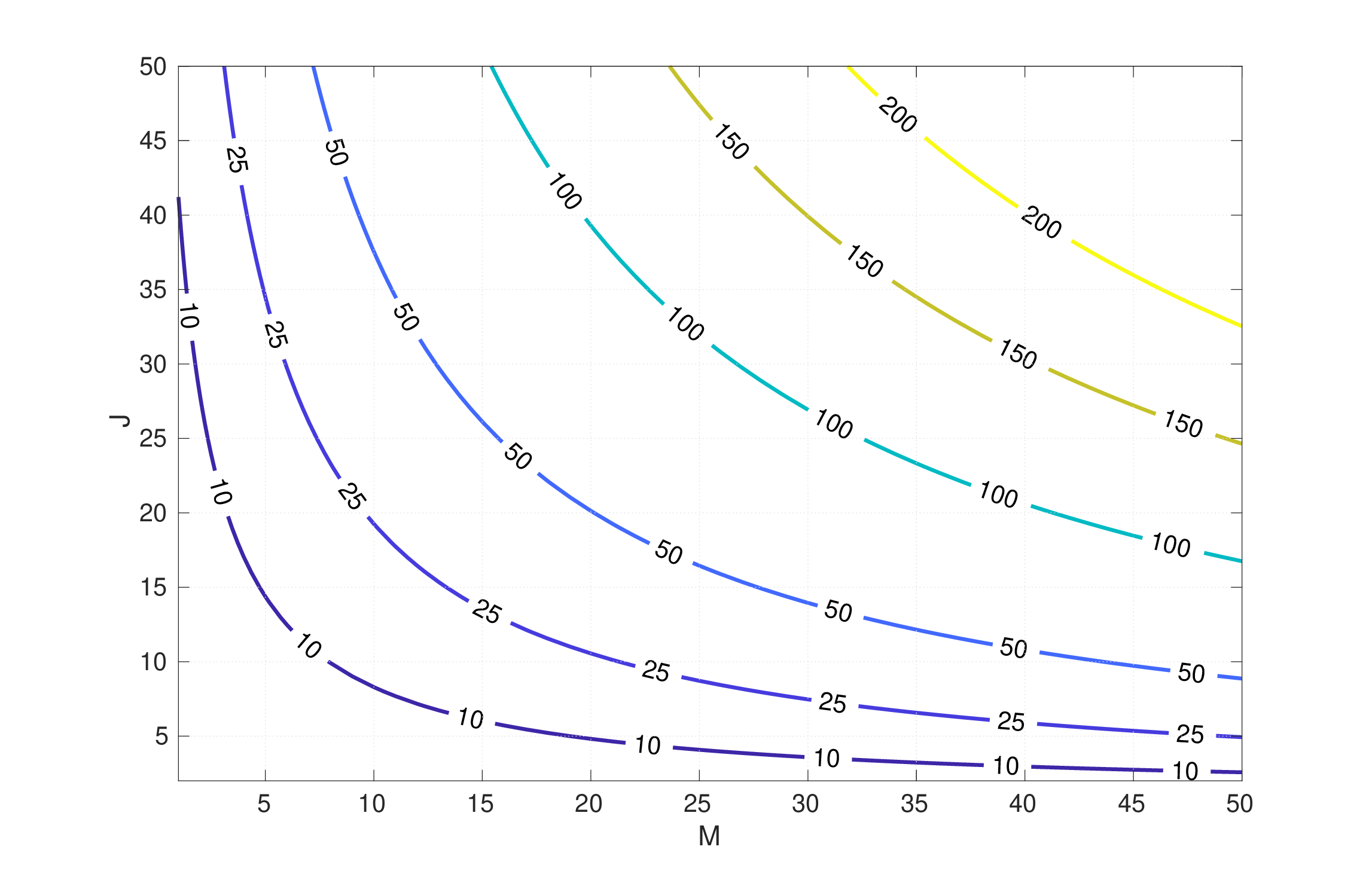}}
\caption{Level curves of $\theta^*$ in function of $M$ and $J$ for fixed $N=300$ and $\gamma=10$. The bottom left corner, corresponding to $M=1$ and $J=2$, is the case of \cite{schied2018market}.}
\label{fig_theta*_M_J}
\end{figure}
   Finally, we examine how $\theta^*$ 
   depends on $N$ and $\gamma$ for fixed $M$ and $J$, which are $M=J=11$, see Figure \ref{fig_theta*_M_J_fixed} which
   illustrates the related surface\footnote{We also compute the same surface for $M=J=3$ and $M=J=5$, and we obtain similar results, available upon request.}.
         \begin{figure}[!t]
\centering
{\includegraphics[width=0.75\textwidth]{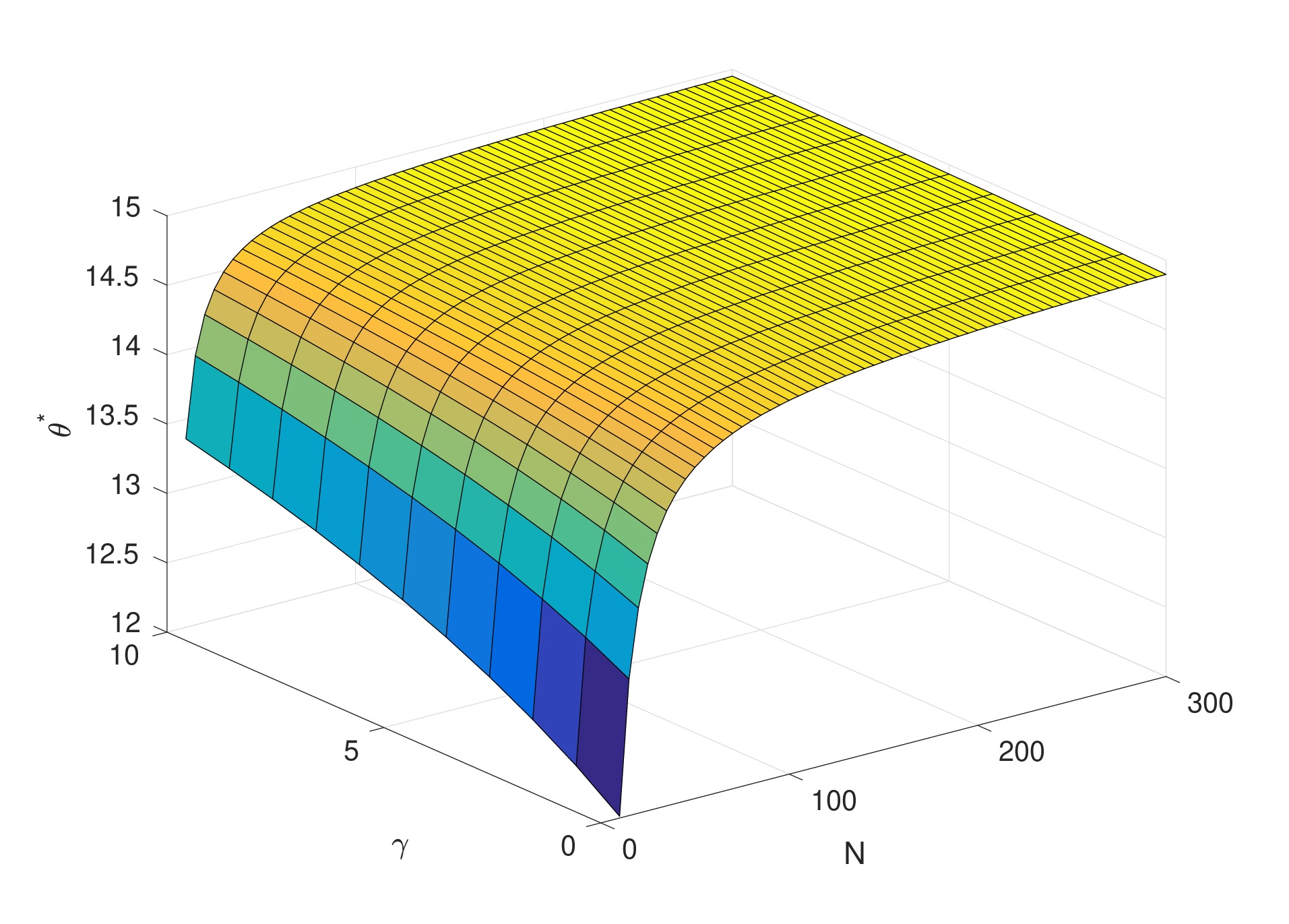}}
\caption{Surface plot of $\theta^*$ in function of $N$ and $\gamma$ for fixed $M=J=11$.}
\label{fig_theta*_M_J_fixed}
\end{figure} 
Overall, the numerical results suggests that for fixed $M$ and $J$ the relation \eqref{eq_conjecture} holds when $N$ is not too small, since for the chosen parameter Eq. \eqref{eq_conjecture} predicts $\theta^*=15$.
    
         


\subsubsection{
Scaling and market stability:
comparative statics for $J$ and $M$}
   \label{scaling_parame}

The conjecture in Eq. \ref{eq_conjecture} indicates that, if market impact does not suitably scale with the number of assets and of agents, market turns out to be unstable for large $J$ and $M$, unless the transaction costs parameter $\theta$ increases appropriately.

From Eq. \ref{eq_conjecture} it is clear that a sufficient condition for asymptotic stability, i.e. that there exists a finite value of $\theta^*$ above which the market is stable in a market with many agents and many assets, is that the two following limits are finite:
\begin{itemize}
    \item $\lim_{M\to \infty} \lambda_{max}$
    \item $\lim_{J \to \infty} JG(0)$
\end{itemize}
The first limit is finite, for example, in a multiasset Kyle model as in \cite{garcia2020multivariate} with a one factor structure for fundamental price and net order flow of noise traders (see Section \ref{sec:scaling}). For the second limit, we have empirical and theoretical evidences that the kernel scales with $J$ as $G(t)=J^{-\beta}\bar G(t)$. Clearly we would need $\beta \ge 1$ to have asymptotic stability in $J$. Multi agent Kyle models with finite variance of the fundamental price suggest $\beta=1/2$ (see \cite{Bagnoli} and Section \ref{sec:scaling}), while the empirical evidence is more ambiguous, due to the difficulty to identify empirically the trading activity of different agents. Thus it is still an open issue if this type of setting can provide asymptotically stable markets in the limit of large $J$. 

\cite{schied2018market} have highlighted the importance of transaction cost $\theta$ in determining whether the Nash equilibrium of a market impact game is stable or unstable. They prove the existence of a threshold value of $\theta$ below which the market is unstable. By extending the framework to many agents and assets, we have shown the importance of the scaling properties of market impact, which in these games is exogenous, in determining market stability. If the largest eigenvalue of the cross impact matrix diverges with $M$ or if the kernel goes to zero slower than $J^{-1}$ the market will be unstable for large values of $M$ or $J$, respectively.

\section{Conclusions}

In this paper we investigated the general problem of Nash equilibria in market impact games with an arbitrary number of assets and of agents. Specifically, we extended the results of \cite{schied2018market} and \cite{luo_schied} in several directions. We first considered a multi-asset market where we introduced the cross-impact effect among assets. We solved the Nash equilibrium, 
we analysed the optimal solution provided by the equilibrium, and we studied the impact of transaction costs
on liquidation strategies. We found that in the presence of cross impact it might be convenient to trade auxiliary assets in order to minimize market impact cost of a position a trader wants to liquidate. Thus the optimal execution problem should be handled intrinsically as a multi-asset problem.

We then used market impact games to investigate several potential determinants of market instabilities driven by finite liquidity and simultaneous trade execution of many agents. In addition to the existence of a transition between a stable and an unstable phase when the transaction cost is smaller than a given threshold (\cite{schied2018market}), we find that markets become asymptotically in $J$ and $M$ unstable when the impact does not scale suitably with these variables. On one side, these results set limits to the parameters of models that do not lead to instability and on the other, contributes to the theoretical and empirical literature on market impact in multi agent and multi-asset settings.

From the policy perspective, the conjecture above indicates that the critical 
transaction cost level $\theta$ below which instabilities are present grows with the impact coefficient $G(0)$ (or its scaled version), the number of traders $J$, and the largest eigenvalue $\lambda_{max}$ of the cross impact matrix. 
Thus, to ensure stability, the transaction cost parameter $\theta$ should be set taking into account the above variables, and be increased or decreased when they significantly change\footnote{In principle, regulators could also act on $G(0)$ by implementing measures making the market more liquid to individual trades, for example modifying the cost of limit orders.}.

Clearly, an increase of the transaction costs might discourage trading activity, therefore decreasing overall market participation and possibly price discovery. For example, in the one period multi-agent Kyle model of \cite{Bagnoli} the mean square deviation of the market price from the fundamental value goes to zero with the number of agents as $(J+1)^{-1}$. Thus regulators should fix transaction costs by balancing the contrasting objectives of increasing traders participation/price discovery and stabilizing markets.

       \clearpage

    \bibliography{bib2}

\clearpage       
\clearpage 

\begin{appendices}
\titlelabel{\appendixname\ \thetitle.\quad}
\appendixtitleon
\section{Proofs of the results}\label{app_1}

 \begin{proof}[Proof of Lemma \ref{lem_NE_diagonal_multi_agent}]
    Since the cross-impact matrix is diagonal, 
each asset is not affected by the orders on other assets, i.e., the impact for each asset is provided only by the 
self-impact and there is no cross-impact effect.
 In particular,
$$
C_{\T}(\Xi_{\cdot,j,\cdot}|\Xi_{\cdot,-j,\cdot})=
\sum_{i=1}^M  C_{\T}(\bm{\xi}_{i,j,\cdot}|\Xi_{i,-j,\cdot}; G_i),
$$
where $C_{\T}(\bm{\xi}_{i,j,\cdot}|\Xi_{i,-j,\cdot}; G_i)$
is the liquidation cost of Definition \ref{de_cost_luo} where the decay kernel is multiplied by $\lambda_i$.
Moreover, the mean-variance functional can be splitted in the sum 
of mean-variance functionals of each asset $i$, i.e.,
\[
MV_{\gamma}(\Xi_{\cdot,j,\cdot}|\Xi_{\cdot,-j,\cdot})=
\sum_{i=1}^M MV_{\gamma}(\bm{\xi}_{i,j,\cdot}|\Xi_{i,-j,\cdot}; G_i),
\]
where $MV_{\gamma}(\bm{\xi}_{i,j,\cdot}|\Xi_{i,-j,\cdot}; G_i)$ is the mean-variance functional defined in equation \eqref{functional_luo_schied} with the related $C_{\T}(\bm{\xi}_{i,j,\cdot}|\Xi_{i,-j,\cdot}; G_i)$.
Indeed, $$MV_{\gamma}(\Xi_{\cdot,j,\cdot}|\Xi_{\cdot,-j,\cdot})=
\E[C_{\T}(\Xi_{\cdot,j,\cdot}|\Xi_{\cdot,-j,\cdot})]+\frac{\gamma}{2}
\mbox{Var}[C_{\T}(\Xi_{\cdot,j,\cdot}|\Xi_{\cdot,-j,\cdot})]$$ and
since $\E[\cdot]$ is a linear operator
$$E[C_{\T}(\Xi_{\cdot,j,\cdot}|\Xi_{\cdot,-j,\cdot})]=
\sum_{i=1}^M E[C_{\T}(\bm{\xi}_{i,j,\cdot}|\Xi_{i,-j,\cdot}; G_i)].$$
On the other hand, 
$\mbox{Var}[C_{\T}(\Xi_{\cdot,j,\cdot}|\Xi_{\cdot,-j,\cdot})]=
\mbox{Var}[\sum_{k=0}^{N} 
\< \bm{S}_{t_k}^\Xi, \bm{\xi}_{\cdot,j,k}\rangle]$ because $\Xi$ is deterministic. Let us denote
$Y_i=\sum_{k=0}^N S_{t_k,i}^\Xi \xi_{i,j,k}$, then
\begin{equation*}
\mbox{Var}\bigg[\sum_{k=0}^{N} 
\< \bm{S}_{t_k}^\Xi, \bm{\xi}_{\cdot,j,k}\rangle\bigg]=\mbox{Var}\bigg[\sum_{i=1}^M Y_i\bigg]=
\sum_{i=1}^{M}\mbox{Var}(Y_i)+\sum_{i\neq l}\mbox{Cov}(Y_i,Y_l).
\end{equation*}
However, if $\mbox{Cov}(Y_i,Y_l)=0$ for $i\neq l$, then $$\mbox{Var}[C_{\T}(\Xi_{\cdot,j,\cdot}|\Xi_{\cdot,-j,\cdot})]=
\sum_{i=1}^M \mbox{Var}[C_{\T}(\bm{\xi}_{i,j,\cdot}|\Xi_{i,-j,\cdot}; G_i) ],$$ 
where we used again that $\Xi$ is deterministic.
Therefore, the $M$ multi-asset 
market impact game with $J$ agents is equivalent to consider $M$ stacked
independent one-asset 
market impact game with $J$ agents, 
where the decay kernel
for each asset $i$  is scaled by the corresponding diagonal element of $D$, $\lambda_i$, which preserves the strictly positive definite property
since $\lambda_i>0\ \forall i$.
Thus,
for each asset $i$ and agent $j$ the existence, uniqueness and the closed formula of Nash Equilibrium $\bm{\xi}^*_{i,j,\cdot}$ for the mean-variance optimization
are straightforward 
from
Theorem 2.4 of \cite{luo_schied} where the decay kernel is multiplied by $\lambda_i$, respectively for each asset.
Moreover, since 
$MV_{\gamma}(\Xi_{\cdot,j,\cdot}|\Xi_{\cdot,-j,\cdot})=
\sum_{i=1}^M MV_{\gamma}(\bm{\xi}_{i,j,\cdot}|\Xi_{i,-j,\cdot}; G_i)$ we may conclude.
If $\bm{S}_{\cdot}^0$ follows a Bachelier model and $\Xi$ is deterministic, then
$C_{\T}(\Xi_{\cdot,j,\cdot}|\Xi_{\cdot,-j,\cdot})$ is a  Gaussian random variable, so that the mean-variance optimization and CARA expected utility maximization are equivalent over the class of deterministic strategies, indeed
$$ U_{\gamma}(\Xi_{\cdot,j,\cdot}|\Xi_{\cdot,-j,\cdot})=
u_{\gamma}(-MV_\gamma(\Xi_{\cdot,j,\cdot}|\Xi_{\cdot,-j,\cdot})),
\quad \gamma >0,
$$
$$ \text{and } U_{0}(\Xi_{\cdot,j,\cdot}|\Xi_{\cdot,-j,\cdot})=
-\E[C_{\T}(\Xi_{\cdot,j,\cdot}|\Xi_{\cdot,-j,\cdot})], \quad \gamma =0.
$$
On the other hand, following the same reasoning of the proof of Theorem 2.4 of \cite{luo_schied}, when $\Xi_{\cdot,-j,\cdot}$ are deterministic,
from Theorem 2.1 of \cite{schied2010optimal}
if there exists a deterministic strategy $\Xi_{\cdot,j,\cdot}^*$ which maximizes 
the expected utility functional
$U_{\gamma}(\Xi_{\cdot,j,\cdot}|\Xi_{\cdot,-j,\cdot})$, 
over the class of deterministic strategies, 
then $\Xi^*_{\cdot,j,\cdot}$ is also a maximizer for 
the expected utility functional
 within the class  of all adapted strategies.
Then, we may use the same argument of Corollary 2.3 of \cite{schied2017state} to conclude that 
the Nash equilibrium for the mean-variance optimization problem form a Nash equilibrium for CARA expected utility maximization.

So, it remains to show that if $\bm{S}_{\cdot}^0$ has uncorrelated components, then 
$\mbox{Cov}(Y_i,Y_l)=0$ for $i\neq l$, where $Y_i=\sum_{k=0}^N S_{t_k,i}^\Xi \xi_{i,j,k}$.
However, ${S}_{t,i}^{\Xi}={S}_{t,i}^0 -\sum_{t_k<t} G(t-t_k)\cdot \sum_{j=1}^J (Q
 \cdot \bm{\xi}_{\cdot,j,k})_i$, where $(Q
 \cdot \bm{\xi}_{\cdot,j,k})_i$ denotes the $i$-th component of $Q\cdot \bm{\xi}_{\cdot,j,k}$, then
 \[ 
 \begin{split}
 Y_i&=\sum_{k=0}^N \bigg[{S}_{t_k,i}^0 \xi_{i,j,k}-\bigg(\sum_{t_k<t} G(t-t_k)
 \cdot \sum_{j=1}^J (Q
 \cdot\bm{\xi}_{\cdot,j,k})_i\bigg) \xi_{i,j,k}\bigg]\\
 &=\sum_{k=0}^N \bigg[{S}_{t_k,i}^0 \xi_{i,j,k}\bigg]-
 \sum_{k=0}^N\bigg[
 \bigg(\sum_{t_k<t} G(t-t_k)
 \cdot \sum_{j=1}^J (Q
 \cdot\bm{\xi}_{\cdot,j,k})_i\bigg) \xi_{i,j,k}\bigg]
 \end{split}
 \]
 so since $\Xi$ is deterministic and using the martingale property of $\bm{S}_{\cdot}^0$,
\[
\begin{split}
\mbox{Cov}(Y_i,Y_l)&=
\mbox{Cov}\bigg(\sum_{k=0}^N S_{t_k,i}^0 \xi_{i,j,k},\sum_{h=0}^N S_{t_h,l}^0 \xi_{l,j,h}\bigg)=\\
&=\E\bigg[ \sum_{k,h=0}^N S_{t_k,i}^0 S_{t_h,l}^0 \xi_{i,j,k} \xi_{l,j,h}\bigg]-\E\bigg[ \sum_{k=0}^N S_{t_k,i}^0 \xi_{i,j,k}\bigg]
\E\bigg[\sum_{h=0}^N S_{t_h,l}^0  \xi_{l,j,h}\bigg]\\
&=\sum_{h,k=0}^N \xi_{i,j,k} \xi_{l,j,h} \mbox{Cov}(S_{t_k,i}^0,S_{t_h,l}^0)=\sum_{h,k=0}^N \xi_{i,j,k} \xi_{l,j,h} \mbox{Cov}(S_{{t_k\wedge t_h},i}^0,S_{{t_k\wedge t_h},l}^0)
\end{split}
\]
which is zero if the components of $\bm{S}_{\cdot}^0$ are uncorrelated.
  \end{proof}

\begin{proof}[Proof of Theorem \ref{te_NE_multi_asset_multi_agent}]

  Let $Q=VDV^{T}$ be the spectral decomposition of $Q$, where, since $Q$ is symmetric,
  $V$ is orthogonal and $D$ is the
diagonal matrix which contains the eigenvalues of $Q$.
  By 
  Assumptions \ref{ass_cross_impact} $\mbox{Cov}(\bm{P}_t^0)=V^T \Sigma V$ is diagonal, so by Lemma \ref{lem_NE_diagonal_multi_agent}
there exists the Nash Equilibrium $\Xi^{*,P}\in 
  \mathscr{X}_{\det}(X^P,\T)$,
for each inventory $X^P$  associated to the orthogonalized virtual assets $\bm{P}_{t}=V^T \bm{S}_{t}$.
Moreover, if $\bm{S}_t^{0}$ follows a Bachelier model then also 
$\bm{P}_t^0$ follows a Bachelier model and
$\Xi^{*,P}$ is 
also a Nash equilibrium for the CARA expected utility maximization for Lemma \ref{lem_NE_diagonal_multi_agent}.
Therefore, to proof that $ \Xi^{*}$, where
$ \Xi_{\cdot,j,\cdot}^{*}= V\Xi_{\cdot,j,\cdot}^{*,P}$, is the Nash
Equilibrium is sufficient to show that the liquidation cost 
$C_{\T}({\Xi_{\cdot,j,\cdot}}|\Xi_{\cdot,-j,\cdot})$, when the cross impact matrix is $Q$,
is equivalent to 
$C_{\T}({\Xi_{\cdot,j,\cdot}^P}|\Xi_{\cdot,-j,\cdot}^P)$, when the cross impact is $D$, 
where the equivalence map is provided by $V^T.$
Writing explicitly for each trading time step $k$ the liquidation cost formula we have, since $V$ is orthogonal,

\[
\begin{split}
  C_{\T}({\Xi_{\cdot,j,\cdot}}|\Xi_{\cdot,-j,\cdot})]&=
  \sum_{k=0}^{N} \bigg (
 \frac{G(0)}{2}\<Q\bm{\xi}_{\cdot,j,k},\bm{\xi}_{\cdot,j,k}\rangle- 
 \<\bm{S}_{t_k}^{\Xi}, \bm{\xi}_{\cdot,j,k}\rangle+\\&+
 \frac{G(0)}{2}\sum_{l\neq j}
 \<Q\bm{\xi}_{\cdot,l,k},\bm{\xi}_{\cdot,j,k}\rangle+\theta\ \< \bm{\xi}_{\cdot,j,k},\bm{\xi}_{\cdot,j,k}    \rangle.
 \bigg)\\
 \end{split}
\]
\[\begin{split}
 &=\sum_{k=0}^{N} \bigg (
 \frac{G(0)}{2}\<DV^T\bm{\xi}_{\cdot,j,k},V^T\bm{\xi}_{\cdot,j,k}\rangle- 
 \<V^T\bm{S}_{t_k}^{\Xi},V^T \bm{\xi}_{\cdot,j,k}\rangle+\\&+
 \frac{G(0)}{2}\sum_{l\neq j}
 \<DV^T\bm{\xi}_{\cdot,l,k},V^T\bm{\xi}_{\cdot,j,k}\rangle+\theta\ \< V^T\bm{\xi}_{\cdot,j,k},V^T\bm{\xi}_{\cdot,j,k}    \rangle.
 \bigg)\\
  &=\sum_{k=0}^{N} \bigg (
 \frac{G(0)}{2}\<D\Xi^P_{\cdot,j,k}, \Xi^P_{\cdot,j,k}\rangle- 
 \< \bm{P}_{t_k},  \Xi^P_{\cdot,j,k}\rangle+\\&+
 \frac{G(0)}{2}\sum_{l\neq j}
 \<D \Xi^P_{\cdot,l,k}, \Xi^P_{\cdot,j,k}\rangle+\theta\ \<  \Xi^P_{\cdot,j,k}, \Xi^P_{\cdot,j,k}    \rangle.
 \bigg)\\
 &=C_{\T}({\Xi_{\cdot,j,\cdot}^P}|\Xi_{\cdot,-j,\cdot}^P).
\end{split}
\]
Finally,
in order to obtain that $\Xi^{*}$ is admissible for $X$, it is sufficient to 
  set $X^P=V^T X$ .
\end{proof}

   \begin{proof}[Proof of Corollary \ref{cor_riks_neutrals}]
     As observed in Remark \ref{os_risk_neutral_agents} the mean-variance functional is splitted as the sum of 
     mean-variance functionals of each asset $i$, since when $\gamma=0$
     the functional is restricted to the expected cost. 
     Then, the existence of the Nash equilibrium for the virtual 
     orthogonalized assets follows by Lemma \ref{lem_NE_diagonal_multi_agent} without requiring the assumptions of uncorrelated assets and the proof follows directly by the same reasoning of the proof of Theorem \ref{te_NE_multi_asset_multi_agent}.
    Moreover by definition, when $\gamma=0$ the CARA utility function is equal to the mean-variance functional, so that $\Xi^*$ is a Nash equilibrium over the set $\mathscr{X}(X,\T)$. 
   \end{proof}

\begin{proof}[Proof of Proposition \ref{pr_kick_out_arbi}]
Let the $j$-th trader be an Arbitrageur, i.e., $\bm{X}_{\cdot,j}=\bm{0}\in \R^M$.
Moreover, his/her inventory for the virtual assets is zero, $X_{i,j}^P= \sum_{m=1}^M V_{i,m}^T X_{m,j}=0$ for each $i=1,2,\ldots,M.$ Then, since
 for Theorem \ref{te_NE_multi_asset_multi_agent} Eq. \eqref{eq_NASH_multi_asset_agent_virtual}
 provides the optimal schedule 
 on each virtual assets $i$, the 
 optimal schedule of the
 Arbitrageur for the $i$-th virtual asset is characterized by the corresponding $\overline{X}_{i,\cdot}^P$.
 
 a) $\Rightarrow$ b).
     If  $\overline{X}_{i,\cdot}=0, \ \forall i$ then 
 $$\overline{X}_{i,\cdot}^P=\frac{1}{J}\sum_{j=1}^J X_{i,j}^P=
 \frac{1}{J}\sum_{j=1}^J \sum_{m=1}^M V_{i,m}^T X_{m,j}=
 \sum_{m=1}^M V_{i,m}^T \overline{X}_{m,\cdot}=0 \quad \forall i.$$
 So, the solution of the Arbitrageurs for each virtual assets is zero and hence also for the original assets by Theorem \ref{te_NE_multi_asset_multi_agent}.

    b) $\Rightarrow$ a).
    If the optimal solution for an Arbitrageur is zero for all assets, then by Theorem \ref{te_NE_multi_asset_multi_agent} and since $V$ is orthogonal, the optimal solution
    for the Arbitrageur is zero also for the virtual assets, so that $\overline{X}_{i,\cdot}^P=0 \ \forall i$ and then
    $\overline{X}_{i,\cdot}=0 \ \forall i$.
\end{proof}
 
 \begin{proof}[Proof of Theorem \ref{te_inst}]

%
Let $\bm{X}_1,\bm{X}_2$ be the inventories of 
trader first and second trader, respectively.
In order to show that market is unstable it is sufficient
to exhibit initial inventories which leads to 
optimal trading strategies with spurious oscillations.
WLOG we may assume that inventories are normalized to $1$, i.e., $\bm{X}_1^T \bm{X}_1=\bm{X}_2^T \bm{X}_2=1.$
 Therefore, let us consider $\bm{X}_1=-\bm{X}_2$,
so that $\bm{X}_1^P=V^T\bm{X}_1=-V^T\bm{X}_2=
-\bm{X}_2^P$ and the NE for the $i$-th virtual assets is 
fully characterized by the fundamental solutions $\bm{w}_i$.
So, for each 
virtual asset the instability is lead by the correspondent virtual kernel, i.e., the  kernel relative to the $i$-th virtual asset which is given by 
$G\cdot \lambda_i$, where $\lambda_i$ is the 
related $i$-th eigenvalues. Then, for the Schied and Zhang instability result we know 
that if we want non oscillatory solutions, $\theta$ has to be greater than $G(0)\cdot \lambda_i /4$ for all $i$. 
However, 
if $\bm{\nu}_i$ denotes the $i$-th eigenvector of $Q$, which may be assumed normalized $\bm{\nu}_i^T\bm{\nu}_i=1$, then 
when $\bm{X}_1=\bm{\nu}_i$ the virtual inventory  
$\bm{X}_1^P$ has $1$ in the $i$-th component and zero otherwise. Then, $\Xi^{*,P}$ is a matrix where the 
$i$-th row is equal to $\bm{w}_i^T$ and zero otherwise.
Therefore,  
\[
\Xi^{*}=V\cdot \Xi^{*,P}=
\begin{bmatrix}
\bm{\nu}_1 |  \cdots
\bm{\nu}_{i-1}|\bm{\nu}_i|\bm{\nu}_{i+1}|\cdots
|\bm{\nu}_M
\end{bmatrix}\cdot \Xi^{*,P}=
\begin{bmatrix}
\nu_{1,i}\bm{w}_i^T \\
\vdots \\
\nu_{M,i} \bm{w}_i^T
\end{bmatrix}=\bm{\nu}_{i} \otimes \bm{w}_i,
\]
i.e. the NE for the $j$-asset is 
given by $v_{j,i} \bm{w}_i$, so also the stability 
for the original asset $\bm{S}_t$ is characterized by
 $\bm{w}_i$. Then, if
  $\theta< \theta^* =\max_{i=1,2,\ldots,M} \frac{G(0)\cdot
\lambda_i}{4}$ and $i_{\max}$ denotes the position of 
the maximum eigenvalue, the NE for inventories
$\bm{X}_1=-\bm{X}_2=\bm{\nu}_{i_{\max}}$ exhibits spurious 
oscillations. 

\end{proof}

\begin{proof}[Proof of Corollary \ref{co_1fac}]
   The eigenvalues of $Q$ are 
   $\lambda_1= 1-q+qM$ and $\lambda_{2:M}=1-q$, 
   where $\bm{v}_1=\bm{e}$, the vector with all 
   1, is the virtual asset associated with $\lambda_1$. Then, when 
   $qM\to \infty$ the first eigenvalue diverges so
   for Theorem \ref{te_inst} we conclude.
   \end{proof}
   
   \begin{proof}[Proof of Corollary \ref{cor_block}]
We first note that by Theorem \ref{te_inst} it is sufficient to prove that there exists a cluster which is unbounded. Indeed, we observe that 
 \[
                            Q=
                            \widehat{Q}+
                              q\begin{bmatrix}
                              \bm{e}_1 \\ 
                              \bm{e}_2\\ \vdots \\
                              \bm{e}_K
                              \end{bmatrix}
                              \begin{bmatrix}
                              \bm{e}_1 & 
                              \bm{e}_2& \cdots &
                              \bm{e}_K
                              \end{bmatrix}
                            \]
                            where
                            \[
                           \widehat{Q}= \begin{bmatrix}
      Q_1-q \bm{e}_1\bm{e}_1^T &  0 & \cdots & 0\\
      0 &  Q_2-q \bm{e}_2\bm{e}_2^T & \cdots & 0\\
      \vdots & & \ddots &  \vdots\\
      0 & \cdots & 0 &Q_{K}-q \bm{e}_K\bm{e}_K^T\\
                              \end{bmatrix}.
                           \]
                           %
                           %
                          Then by Theorem 8.1.8 pag.443 of \cite{golub13} 
                          $\lambda_1(Q)\geq \lambda_1(\widehat{Q})$
                           where $\lambda_i(Q)$ denotes the $i$-th largest eigenvalue of $Q$ and 
                           respectively of $\widehat{Q}$.
                          The eigenvalues of $\widehat{Q}$ are given 
                           by the eigenvalues
                           of $Q_i-q \bm{e_i}\bm{e_i}^T$ for $i=1,2,\ldots,K$
                           and for each $i$,
                           $\lambda_{1}(Q_i-q \bm{e_i}\bm{e_i}^T)=1-q_i+M_i(q_i-q)$ and the rests $M_i -1$ eigenvalues are equal to $1-q_i$. 
                           So, if there exists a cluster such that $M_i$ is unbounded for any value of $\theta$, then $\lambda_1(Q_i-q\bm{e}_i\bm{e}_i^T)$ is unbounded, thanks to (i), and also the respective eigenvalue of $Q$, so 
                           by 
                           Theorem \ref{te_inst} we conclude that there is no a finite value for $\theta$ such that the market is weakly stable.
                            
                           So, let us first start by fixing the number of cluster to $K<\infty$.
                            Then, when $M$ tends to infinity 
    at least one of the clusters will increase
    to infinity, which means that 
    there exists a cluster $i$ such that $\lambda_{1}(Q_i-q \bm{e_i}\bm{e_i}^T)\to \infty$, thanks to (i), and also
  the respective eigenvalue of $Q$ goes to infinity.
    Therefore, we conclude for Theorem \ref{te_inst}.
    
   For the general case we conclude by contradiction. If
    $K(M)$ is the number of cluster for a fixed $M$, and $K(M)\to \infty$ when $M\to \infty$
    then the set $\{M_i : i \in \N \}$ is unbounded. Indeed, if 
    $\sup_{i\in \N}{M_i}=S<\infty$, then 
    the average number of stocks in a cluster is $\frac{\sum_{i=1}^{K(M)}M_i }{K(M)}\leq 
    S$ for all $M$ and this is in contradiction with the assumptions that 
    $\lim_{M\to+ \infty}\frac{M}{K(M)}\to +\infty$.
    So since $\{M_i : i \in \N \}$ is unbounded we conclude that there is no finite value of $\theta$ such that it is greater than all the eigenvalues of $Q$ when $M\to \infty.$
\end{proof}

\begin{proof}[Proof of Theorem \ref{te_one_factor_largest}]
  The largest eigenvalue of a symmetric $M\times M$ matrix $Q$ can be defined as
$$
\lambda_{1}(Q)=\max_{\bm{x}\ne 0}\frac{\bm{x}^T Q \bm{x}}{\bm{x}^T \bm{x}}.
$$
If we consider the vector $\bm{e}=(1,1,...,1)^T$, we have the lower bound
$$
\lambda_{1}(Q)\ge \frac{\bm{e}^T Q \bm{e}}{\bm{e}^T \bm{e}}=\frac{\sum_{i,j}q_{ij}}{M}.
$$
The largest eigenvalue of a generic matrix $Q \in {\cal A}_h^M$ is then bounded by
$$
\lambda_{1}(Q)\ge 1+ \frac{2h}{M}.
$$
But the one-factor matrix $Q_{1fac}=(1-q)I_M+q \bm{e}\bm{e}^T$, 
with $q= \frac{2h}{M(M-1)}$, belongs to ${\cal A}^M_h$ and has 
$$
\lambda_{1}(Q_{1fac})=1+(M-1)\frac{2h}{M(M-1)}= 1+ \frac{2h}{M},
$$
i.e., the lower bound for the max eigenvalue of matrices in ${\cal A}^M_h$.
Therefore, $\forall Q \in {\cal A}^h$ it holds that
$$
\lambda_{1}(Q) \ge \lambda_{1}(Q_{1fac}).
$$
 \end{proof}

Note that the bound is not strict since the largest eigenvalue of a block diagonal matrix with identical blocks is also $1+ \frac{2h}{M}$. Indeed,
let consider the block diagonal matrix with $K$ identical clusters
   \[
                            Q:=\begin{bmatrix}
      Q(\rho) &  0 & \cdots & 0\\
     0 & Q(\rho) & \cdots &0\\
      \vdots & & \ddots &  \vdots\\
      0 & \cdots & 0 & Q(\rho)\\
                              \end{bmatrix}\in \R^{M\times M},
                           \]
    where $Q(\rho)\in \R^{M_c}$ is a one-factor matrix and $M_c \cdot K=M$. 
    We observe that $Q\in \mathcal{A}_h^M$ if and only if $\rho=\frac{2h}{(M_c -1)M}$, therefore \[
    \lambda_1(Q)=1+(M_c-1) \rho=1+\frac{2h}{M}.
    \]

\end{appendices}

    \clearpage


\end{document}